\documentclass{article}

\usepackage[utf8]{inputenc}
\usepackage{lmodern} 
\usepackage[bookmarks,colorlinks,breaklinks]{hyperref}
\hypersetup{urlcolor=blue, colorlinks=true, citecolor=green!50!black, linkcolor=blue}
\usepackage{multirow}
\usepackage[T1]{fontenc} %
\usepackage{amsmath}
\usepackage{amsfonts}
\usepackage{amssymb}
\usepackage{amsthm}
\usepackage{thmtools}
\usepackage{pgffor}
\usepackage{xcolor}
\usepackage[lined,boxed,ruled,norelsize,algo2e,linesnumbered,noend]{algorithm2e}
\usepackage{cleveref}
\usepackage{graphicx}
\usepackage{geometry}
\usepackage{enumitem}
\usepackage{url}
\usepackage{tikz}
\usetikzlibrary{quotes}

\declaretheorem[numberwithin=section,refname={Theorem,Theorems},Refname={Theorem,Theorems}]{theorem}
\declaretheorem[numberlike=theorem]{lemma}

\declaretheorem[numberlike=theorem]{corollary}
\declaretheorem[numberlike=theorem]{definition}
\declaretheorem[numberlike=theorem]{property}

\declaretheorem[numberlike=theorem]{fact}

\newcommand{\R}{\mathbb{R}}

\newcommand{\N}{\mathbb{N}}

\newcommand{\F}{\mathbb{F}}

\renewcommand{\P}{\mathbb{P}}

\renewcommand{\tilde}{\widetilde}
\renewcommand{\hat}{\widehat}

\newcommand{\rank}{\operatorname{rank}}
\newcommand{\poly}{\operatorname{poly}}
\newcommand{\polylog}{\operatorname{polylog}}
\newcommand{\sign}{\operatorname{sign}}%
\newcommand{\nnz}{\operatorname{nnz}}%

\foreach \x in {A,...,Z}{%
	\expandafter\xdef\csname m\x\endcsname{\noexpand\mathbf{\x}}
}

\foreach \x in {A,...,Z}{%
	\expandafter\xdef\csname om\x\endcsname{\noexpand\overline{\noexpand\mathbf{\x}}}
}

\foreach \x in {A,...,Z}{%
	\expandafter\xdef\csname c\x\endcsname{\noexpand\mathcal{\x}}
}

\newcommand{\otau}{\noexpand{\overline{\tau}}}
\newcommand{\omu}{\noexpand{\overline{\mu}}}
\newcommand{\osigma}{\noexpand{\overline{\sigma}}}%
\newcommand{\os}{\noexpand{\overline{s}}}
\newcommand{\ot}{\noexpand{\overline{t}}}
\newcommand{\og}{\noexpand{\overline{g}}}
\newcommand{\ou}{\noexpand{\overline{u}}}
\newcommand{\ov}{\noexpand{\overline{v}}}
\newcommand{\ow}{\noexpand{\overline{w}}}
\newcommand{\ox}{\noexpand{\overline{x}}}
\newcommand{\oy}{\noexpand{\overline{y}}}
\newcommand{\oz}{\noexpand{\overline{z}}}
\newcommand{\oc}{\noexpand{\overline{c}}}

\makeatletter
\renewcommand{\paragraph}{%
	\@startsection{paragraph}{4}%
	{\z@}{1.25ex \@plus 1ex \@minus .2ex}{-1em}%
	{\normalfont\normalsize\bfseries}%
}
\makeatother

\makeatletter
\renewcommand{\subparagraph}{%
	\@startsection{paragraph}{4}%
	{\z@}{1.25ex \@plus 1ex \@minus .2ex}{-1em}%
	{\normalfont\normalsize\bfseries}%
}
\makeatother

\def\arxiv{1}
\ifdefined\DEBUG
	\usepackage{showlabels}

\else
\fi

\newcommand{\sankowski}{{1.446}}
\newcommand{\bns}{{1.405}}
\newcommand{\columnUpdate}{{1.528}}

\title{Dynamic Rank, Basis, and Matching}

\author{Jan van den Brand\footnote{Georgia Institute of Technology.} \and Vishal Kumar\footnote{Teachers College, Columbia University. Work done while at Georgia Institute of Technology.} \and Daniel J. Zhang\footnote{Georgia Institute of Technology.}}
\date{}

\begin{document}

\pagenumbering{roman}
\maketitle

\begin{abstract}
    
We study dynamic algorithms for maintaining fundamental algebraic properties of matrices, specifically, rank, basis, and full-rank submatrices, with applications to maximum matching on dynamic graphs.
Prior dynamic algorithms for rank achieve subquadratic update times but scale with the matrix dimension $n$, and could not always maintain the corresponding objects such as a basis or maximum full-rank submatrix.

We present the first dynamic rank algorithms whose update time scales with the matrix rank 
$r$, achieving $\tilde O(r^\bns)$ time per entry-update and $\tilde O(r^\columnUpdate + z)$ per column-update, where 
$z$
is the number of changed entries. 
This extends to $\tilde O(|M|^\bns)$ edge-update time to maintain the size $|M|$ of a maximum matching.
We also give %
dynamic algorithms for maintaining a column-basis subject to column-updates and a maximum full-rank submatrix subject to entry-updates.

\end{abstract}

\tableofcontents

\clearpage

\pagenumbering{arabic}

\section{Introduction}
\label{sec:intro}

Dynamic algebraic algorithms maintain properties of matrices and vector spaces while the input undergoes changes. 
These algorithms, when used as subroutines, have led to improvements in various areas such as convex optimization \cite{CohenLS21,Brand21,JiangKLP020,JiangLSW20}, network-flows \cite{DongGGLPSY22,GaoLP21,BrandLLSSSW21,BrandGJLLPS22,BrandZ23},
online algorithms \cite{DengSW22,HarshawSSZ19},
failure-sensitive oracles \cite{GrandoniV20,WeimannY13,BrandS19,GuR21}.

Many dynamic graph problems (where the goal is to maintain certain graph properties while the graph undergoes edge insertions and deletions) such as
reachability \cite{Sankowski04,GoranciKMP25,KarczmarzS23c}, strong-connectivity \cite{KarczmarzS23b} shortest path \cite{Sankowski05,KarczmarzS23,BrandK23,BrandN19}, matching \cite{Sankowski07},
near additive spanner \cite{BergamaschiHGWW21,BrandFN22}, 
have also been addressed using dynamic algebraic techniques.
In fact, it has been shown that achieving subquadratic upper bounds for many of these problems requires the use of algebraic techniques. \cite{AbboudW14,BergamaschiHGWW21}.

This work focuses on two of the most fundamental properties of matrices and vector spaces: \emph{rank} and \emph{basis}.
Formally, the problem is defined as follows:
we are given an $n\times n$ matrix (i.e., a collection of vectors) that undergoes updates (e.g., column replacements). After each update, we must return the new rank or basis. The ``\emph{update-time}'' refers to the time required to compute and return the new rank or basis after an update.
We consider both \emph{column-updates} (where an entire column/vector is replaced) and \emph{entry-updates} (where a single matrix entry is modified).

\subparagraph{Dynamic Rank \& Matching:}

Given a graph $G$, the size of a maximum matching is equal to half the rank of the graph's Tutte-matrix\footnote{The Tutte-matrix \cite{tutte1947factorization} is the symbolic adjacency matrix where the signs are flipped along the diagonal, i.e. $\mA_{i,j}=-\mA_{j,i}$.}.
Hence, by maintaining the rank of a dynamic matrix under entry-updates, we can maintain the size of a maximum matching in a graph that undergoes edge insertions and deletions. The latter problem is known as  dynamic maximum matching.
Via further reductions, this also extends to other graph problems such as dynamic vertex-capacitated max-flow.

Dynamic maximum matching is a well-studied problem in the dynamic graph area with substantial amount of work in recent years for the approximate setting \cite{Behnezhad23,Kiss23,BlikstadK23,BhattacharyaKS23,AzarmehrBR24,BhattacharyaKSW24,AssadiKK25}.
In the exact setting, the upper bounds for $n$-vertex graphs are $O(n^{\columnUpdate})$ \cite{FrandsenF09}, $O(n^{\sankowski})$ \cite{Sankowski07} and $O(n^\bns)$ \cite{BrandNS19}.
These upper bounds for exact matching are achieved via the reduction from matching to rank, i.e., despite the graph application, these dynamic algorithms are entirely algebraic and maintain the rank of the Tutte-matrix.
The algebraic approach and slower progress in the exact setting can be explained via conditional lower bounds.
The two highest corresponding conditional lower bounds for dynamic matching are $\Omega(n^{\omega-1}) \sim \Omega(n^{1.372})$ based on the triangle detection conjecture%
\footnote{Triangle conjecture: Given an $n$-vertex graph, triangles cannot be detected in $O(n^{\omega-\epsilon})$ time.} \cite{AbboudW14},
and $\Omega(n^{\bns})$ \cite{BrandNS19} based on the hinted-Mv conjecture%
\footnote{The hinted-Mv conjecture \cite{BrandNS19} is a variation of the OMv conjecture \cite{HenzingerKNS15}, parametrized by the fast matrix multiplication exponent. The lower bound $\Omega(n^{\bns})$ is for current bounds on rectangular matrix multiplication \cite{AlmanDWXXZ25}.}.
These lower bounds imply that any upper bound must be algebraic (i.e., make use of fast matrix multiplication) to beat the naive $\tilde O (n^2)$ upper bound\footnote{We use $\tilde O(\cdot)$ to hide $\polylog(n)$ factors.} of recomputing the matching from scratch on dense graphs \cite{BrandLNPSSSW20,ChenKLPGS23}.
Further, depending on which conjecture one believes, the current $O(n^{\bns})$ update time is conditionally optimal or nearly so (with only a $o(n^{0.04})$ gap to triangle detection for current bounds on $\omega$). While these conditional lower bounds seemingly rule out substantial further progress in the exact setting,
they only holds for distinguishing between matching-size (or rank) $n$ vs $n-1$. 
Thus, the lower bounds do not apply to graphs with small maximum matchings (or matrices with small rank), leaving open the possibility of significant improvements for the low-rank/small-matching case.

Such improvements for low-rank matrices have previously been obtained in the static setting.
While computing the rank of an $n\times n$ matrix $\mA$ generally takes $O(n^\omega)$ time, \cite{CheungKL13} showed how to compute the rank in only $\tilde O(\nnz(\mA)+r^\omega)$ time, where $\nnz(\mA)$ is the number of non-zeros (i.e., time to read the input matrix $\mA$) and $r$ is the rank of the input matrix.
So for low-rank matrices (e.g., $r < n^{1/2}$), the rank can be computed in nearly-linear time.

On the one hand, this highlights a gap in the dynamic setting: for sparse low-rank matrices, na\"ive recomputation of the rank from scratch via \cite{CheungKL13} is already faster than the current $O(n^\bns)$ dynamic algorithms. So no non-trivial dynamic algorithm is known for such matrices.
On the other hand, the result of \cite{CheungKL13} suggests that it may be possible to improve the dynamic complexity by having it depend on rank $r$ rather than dimension $n$.
Our first result shows this is indeed possible: The rank $r$ 
can be maintained in $\tilde O(r^\bns)$ time per entry-update. This also implies a dynamic algorithm for maximum matching size with $O(|M|^\bns)$ edge-update time, where $|M|$ is the size of the maximum matching.

\subparagraph{Dynamic Basis and Full-Rank Submatrix:}

The rank of a set of vectors (or a matrix) is the size of the maximum linearly independent subset (or size of the maximum full-rank submatrix).
Previous fastest dynamic rank data structures \cite{Sankowski07,BrandNS19,BrandFNP24} can maintain the rank (i.e., \emph{size} of basis/submatrix), but not the actual basis or submatrix itself.
These algorithms work via a black-box reduction by \cite{Sankowski07}, which reduces dynamic rank to dynamic matrix inverse (here one must maintain the inverse while the input matrix changes over time).
Thus improvements for dynamic matrix inverse \cite{BrandNS19,BrandFNP24,AnandBGZ24} have directly translated to better algorithms for dynamic rank and matching size.

However, when it comes to maintaining an actual basis, the only known algorithm is that by \cite{FrandsenF09} which, rather than a black-box reduction, was a dynamic algorithm specifically designed for maintaining the rank\footnote{\label{foot:frandsen}\cite{FrandsenF09} state their results for dynamic rank. It was not explicitly stated that they maintain basis or submatrix, but internally their algorithm maintains the positions of leading 1s of the row-echelon form. This allows to extract the basis.}.
Thus, this algorithm has not benefited from recent advances in dynamic matrix inversion, and thus the update time remains at $O(n^{\columnUpdate})$, which is slower than the $O(n^\bns)$ update time for rank. %
If there was a black-box reduction from basis/full-rank submatrix to dynamic matrix inverse, similar to how there is a reduction from rank, then these issues would be resolved.

The ability to maintain the size, rather than the corresponding object, is a common limitation of dynamic algebraic algorithms. This issue extends to the applications of dynamic algebraic algorithms such as maintaining reachability, strong-connectivity, maximum matching, or shortest paths in dynamic graphs.
\cite{AbboudW14} have proven that solving these problems in the dynamic setting faster than the trivial $\tilde O(n^2)$ requires the use of fast matrix multiplication, i.e., algebraic techniques.
However, given the algebraic nature, most dynamic algebraic algorithms maintain numbers/quantities rather than the respective combinatorial objects.
For example,
\cite{Sankowski05,BrandN19,KarczmarzS23,BrandFN22} maintain distances (i.e., length of shortest path) rather than the paths.
\cite{Sankowski07} maintains the size of the maximum matching rather than the matching.
For reachability, \cite{Sankowski04,DemetrescuI05} maintain boolean reachability information, but not the respective paths.
\cite{BrandNS19} maintains a boolean output whether the graph is strongly-connected, but not the strongly-connected components.
This inability to maintain the underlying object is a major draw-back of dynamic algebraic methods and has been stated as open problems in \cite{BrandN19,BrandNS19}.
Recently, studying how to maintain objects via dynamic algebraic techniques has developed to an active area of research with first progress for paths:
\cite{KarczmarzMS22} extended dynamic reachability to path reporting, and \cite{BergamaschiHGWW21,AlokhinaB24} can maintain shortest paths.
\cite{KarczmarzS23b} extends strong-connectivity to maintaining the strongly-connected components.
However, for maintaining a maximum full-rank submatrix (rather than just the rank), or a maximum matching (rather than just its size), the question remains open.

In this work, we give the first reduction from basis (and maximum full-rank submatrix) to dynamic rank.
The latter reduces to dynamic matrix inversion \cite{Sankowski07}, thus, all recent improvements to dynamic matrix inverse now also extend to dynamic basis.
As a full-rank submatrix corresponds to the vertex set used by the maximum matching, this is also the first fully dynamic algorithm that maintains information about the exact maximum matching beyond its size.

\subsection{Our Results}

\begin{figure}
    \centering
    {\renewcommand{\arraystretch}{1.2}
    \begin{tabular}{c|r|l|r|l|}
         & \multicolumn{2}{l|}{Entry Update} & \multicolumn{2}{l|}{Column Update} \\
         Output & New & Previous & New & Previous \\
        \hline
        Rank & ${r^{\bns}}$ & $n^\bns$ \cite{Sankowski07,BrandNS19} & ${z+r^{\columnUpdate}}$ & $n^\columnUpdate$ \cite{Sankowski07,BrandNS19}\\
        Column-Basis & $n^{\bns}$ and $r^{\columnUpdate}$ & $n^\columnUpdate$ \cite{FrandsenF09}$^{\ref{foot:frandsen}}$ & ${z+r^{\columnUpdate}}$  \\
        Submatrix & $n^{\bns}$ and $r^{\columnUpdate}$ 
    \end{tabular}
    }
        \caption{New and previous upper bounds for dynamically maintaining rank, basis, full-rank submatrix. Parameter $n$ bounds the width and height of the input matrix, $r$ is its rank, and $z$ is the number of entries changed by the column update (i.e., input size of the update).}
    \label{fig:results}
\end{figure}

\subparagraph{Maintaining the Rank}

Our first result are improved time complexities for maintaining the rank of a dynamic $n\times n$ matrix, scaling with the rank of the matrix rather than its dimension. In the following, $\nnz(\mA)$ is the number of non-zero entries in $\mA$, when given in sparse representation.

\begin{theorem}\label{thm:dynamicrank}\label{thm:main_rank}
    There are randomized dynamic algorithms that maintain, with high probability, the rank of a dynamic $n\times n$ matrix $\mA$. Both have preprocessing time $\tilde O(\nnz(\mA) + \rank(\mA)^\omega)$ and their update times are:
    \begin{itemize}
        \item $\tilde O(\rank(\mA)^{1.405})$ time per entry-update,
        \item $\tilde O(\rank(\mA)^{1.528}+z)$ time per column-update, where $z$ is the number of entries changed in the column.
    \end{itemize}
    The dynamic algorithms are robust against output-adaptive adversaries\footnote{Here the adversary sees the output of the data structure, but not the internal randomness.}.
\end{theorem}

The term $z$ in the column-update time accounts for reading the input (i.e., change of the column). For dense updates to low-rank matrices, the column-updates are nearly-linear time.
The complexity dependence on $\rank(\mA)$ is conditionally optimal, i.e., any further improvement would also improve the high-rank time complexities of $O(n^\bns)$ and $O(n^\columnUpdate)$.
These are also the first non-trivial dynamic algorithms for low-rank matrices with $\nnz(\mA) < n^{\bns}$ or $n^{1.528}$ respectively.

\Cref{thm:main_rank} is obtained by developing a black-box reduction from maintaining the rank of $n\times n$ matrices to $r\times r$ matrices, where $r = \tilde O(\rank(\mA))$, via techniques of \cite{CheungKL13}.

\begin{restatable}{theorem}{thmSmallRankReduction}\label{thm:rankreduction}
    Assume there is a data structure that maintains the rank of a dynamic $n\times n$ matrix $\mA$ with $U(n,z)$ column-update time, where $z$ is the number of changed entries. Let $P(n)$ be the preprocessing time.

    Then there is a randomized dynamic algorithm for maintaining the rank of $\mA$ after $\tilde O(\nnz(\mA)+P(\rank(\mA)))$ time preprocessing with
    $$\tilde O(z + U(\rank(\mA), z)+P(\rank(\mA))/\rank(\mA))$$
    update time for column updates that change $z$ entries. The output is correct with high probability and the dynamic algorithm is robust against output-adaptive adversaries.
\end{restatable}

Applying this reduction to the best-known dynamic rank algorithms ($O(n^\columnUpdate)$ for column-updates, and $O(n^\bns)$ for entry-updates \cite{Sankowski05,BrandNS19}), results in \Cref{thm:main_rank}.
These upper bounds depend on the current best bounds for fast rectangular matrix multiplication \cite{AlmanDWXXZ25}.

Since the rank of a graph's Tutte-matrix is twice the size of a maximum matching, \Cref{thm:main_rank} directly implies algorithms for dynamic maximum matching size.
For bipartite graphs, a reduction by \cite{KaoLST01} allows us to maintain the weight of the maximum weight matching as well.

\begin{restatable}{theorem}{thmMatching}\label{thm:main_matching}
    There are randomized dynamic algorithms that maintain with high probability
    \begin{itemize}
        \item the weight $k$ of a maximum weight matching in a bipartite graph with $O(W\cdot k^{\bns})$ time per edge-update, where $W$ is the weight of the updated edge.
        \item the size $|M|$ of the maximum matching $M$ in a general graph in $O(|M|^{\bns})$ time per edge-update.
    \end{itemize}
\end{restatable}

The previous best bounds for these problems were $O(n^\bns)$ for maximum matching and $O(W^{2.405}\cdot n^\bns)$ for the weighted case \cite{Sankowski05,BrandNS19}.

\subparagraph{Maintaining a Basis}

The dynamic algorithms of \Cref{thm:main_rank}, and those of \cite{Sankowski05,BrandNS19} can only maintain the rank of a matrix (i.e., set of vectors) but not a corresponding basis.
We extend the results to maintain a column-basis:

\begin{theorem}\label{thm:main_basis}
    There is a randomized dynamic algorithm that maintains with high probability a column-basis of a dynamic $n\times n$ matrix $\mA$ in $\tilde O(z+\rank(\mA)^{1.528})$ time per column update. 
    Here $z$ is the number of entries changed during the column update.
    
    The dynamic algorithm is robust against output-adaptive adversaries.
\end{theorem}

In case of column-updates, the time complexity matches that of \Cref{thm:main_rank} for maintaining the rank. For entry-updates (i.e., when $z=1$) there is a small slow-down compared to the $O(n^\bns)$ update time for maintaining the rank.
It still improves over the previous best bound $O(n^\columnUpdate)$ for maintaining a basis under entry-updates \cite{FrandsenF09}.

Our result is the first to maintain a basis subject to column-updates.
Column-updates to a matrix correspond to adding and removing vectors to some set $V=\{v_1,...,v_n\}$ and maintaining a maximum linear independent subset of them. In this context, column-updates are arguably more natural than entry-updates.

\Cref{thm:main_basis} is obtained by developing a reduction and using the $O(n^{1.528})$ column-update time dynamic rank algorithm from \cite{BrandNS19}.

\begin{restatable}{theorem}{thmBasisReduction}\label{thm:basis_reduction}
    Given a dynamic algorithm for maintaining the rank of a dynamic matrix $\mA \in \F^{n\times n}$ with column-update time $U(n)$ and processing time $P(n)$.

    Then there is dynamic algorithm maintaining a column-basis of $\mA \in \F^{n\times n}$ with column-update time $\tilde O(P(\rank(\mA))/\rank(\mA) + U(\rank(\mA))+z)$. Here $z$ is number of entries that changed in the updated column. 

    The dynamic algorithm is robust against output-adaptive adversaries.
\end{restatable}

\subparagraph{Maintaining a Full-Rank Submatrix}

Given a matrix $\mA \in \F^{n\times n}$, computing a full-rank submatrix corresponds to first constructing a column-basis (i.e., subset of columns of $\mA$) 
and row-basis (i.e., subset of rows of $\mA$), and then taking the submatrix induced by the row and column set\footnote{This only works when we have a column- and row-\emph{basis}. It does not work when these are merely linearly independent column (or row) vectors. 
\ifdefined\arxiv See \Cref{app:basis} for details.
\else See full version for details.
\fi}.

Running \Cref{thm:main_basis} or \cite{FrandsenF09} on $\mA$ and $\mA^\top$ would allow to maintain a maximum full-rank submatrix in $O(\rank(\mA)^{\columnUpdate})$ or $O(n^\columnUpdate)$ time.
However, maintaining the rank can be done in $O(\rank(\mA)^\bns)$ time via \Cref{thm:main_rank} (or $O(n^\bns)$ via \cite{Sankowski07,BrandNS19}).

To address these issues, we develop a new approach for maintaining a maximum full-rank submatrix.
We show that maintaining such a submatrix (and thus also the rank of a matrix) reduces to merely detecting when a dynamic matrix becomes singular (i.e., non-invertible; $\det(\mA)=0$) for the first time.

\begin{restatable}{theorem}{thmSubmatrix}\label{thm:submatrixreduction}
    Assume there is a data structure that detects when a dynamic $n\times n$ matrix $\mA$ becomes singular for the first time, with $U(n)$ time per entry-update. Let $P(n)$ be the preprocessing time.

    Then there is a dynamic algorithm for maintaining a maximum full-rank submatrix (i.e., sets $I,J\subset[n]$ with $\det(\mA_{I,J})\neq 0$ and $|I|=|J|=\rank(\mA)$) after $O(P(n)+n^\omega)$ time preprocessing with
    $\tilde O(U(n))$ entry-update time.

    \noindent
    The reduction is randomized, correct with high probability, and works against an output-adaptive adversary.
\end{restatable}

A reduction from rank to singularity-detection was previously given by \cite{Sankowski07}, but it could only maintain the rank, and not a corresponding object such as a basis or full-rank submatrix. 

By applying our reductions to dynamic determinant algorithm of \cite{BrandNS19}, we obtain the following result:
\begin{theorem}\label{thm:main_submatrix}
    There is a dynamic algorithm that maintains with high probability a $\rank(\mA)\times\rank(\mA)$-sized full-rank submatrix of a dynamic $n\times n$ matrix $\mA$. After each update, the data structure returns the set of row/column-indices that define the submatrix. The preprocessing time is $O(n^\omega)$ and the update time is $\tilde O(n^{1.405})$ per entry update.
\end{theorem}
This result also implies that we can maintain the vertex set of a maximum matching in a dynamic graph undergoing edge updates.
This improves over previous algorithms, which could only maintain the size of the matching \cite{Sankowski07,BrandNS19,BrandFNP24}.

\subparagraph{Remark: Combinatorial Approach}

A combinatorial approach for maintaining a matching (besides recomputation from scratch in $\tilde O(n^2)$ time \cite{BrandLNPSSSW20,ChenKLPGS23}) is to find an augmenting path in $O(n^2)$ time after each update.
Beating quadratic time requires algebraic techniques \cite{AbboudW14}.

Given how there is an $O(n^2)$ combinatorial and $O(n^{1.405})$ algebraic algorithm, it is natural to ask if there is also an $O(|M|^2)$ combinatorial algorithm matching the $O(|M|^{1.405})$ proven in \Cref{thm:main_matching}.
We observe that there is an $O(|M|^2)$ time combinatorial algorithm based on augmenting paths using vertex sparsification from \cite{GuptaP13}.

\ifdefined\arxiv
\begin{restatable*}{theorem}{thmCombinatorial}\label{thm:main_combinatorial}
    There is a deterministic algorithm that maintains the maximum matching $M \subset E$ in a dynamic bipartite graph $G=(V,E)$ subject to edge updates in $O(|M|^2)$ time.
\end{restatable*}
\else
\begin{theorem}{theorem}{thmCombinatorial}\label{thm:main_combinatorial}
    There is a deterministic algorithm that maintains the maximum matching $M \subset E$ in a dynamic bipartite graph $G=(V,E)$ subject to edge updates in $O(|M|^2)$ time.
\end{theorem}
\fi

\subsection{Further Related Work}

In addition to the conditional lower bounds stated in the introduction, there is also an $\Omega(n)$ conditional lower bound for dynamic maximum matching based on the OMv-conjecture \cite{HenzingerKNS15}.

Dynamic matching (maximal, maximum, weighted, etc.) is one of the most well-studied problems in the dynamic graph area, see e.g.,\cite{AssadiBD22,BernsteinDL21,BernsteinFH19,BernsteinS16,BernsteinS15,PelegS16,Wajc20,GuptaP13,LeMSW22,Kiss22,BhattacharyaKS23b,Kiss21,ChechikZ19,CharikarS18,LeMSW22,BhattacharyaK21,BhattacharyaHN16,BehnezhadDHSS19,ArarCCSW18,BaswanaGS11,OnakR10,Behnezhad23,Kiss23,BlikstadK23,BhattacharyaKS23,AzarmehrBR24,BhattacharyaKSW24,AssadiKK25} and references therein.
Maintaining small matching/max-flow in the partially dynamic setting (only insertion/only deletions) was studied in \cite{Henzinger97,GuptaK21,GoranciH23}.
Further dynamic algebraic problems that have been studied are dynamic characteristic polynomial \cite{FrandsenS11}, dynamic Frobenius normal form \cite{KarczmarzS23}, dynamic matrix vector multiplication \cite{AnandBM25},
dynamic convolution (polynomial multiplication) \cite{ReifT97}. 
\cite{FrandsenHM01} studied dynamic algebraic algorithms from the perspective of algebraic circuit lower bounds.

Similar to the exact setting, in the better-than-2 approximate setting the dynamic matching algorithms from
\cite{Behnezhad23,BhattacharyaKSW24,BhattacharyaKS23} 
can also only maintain the size but not the matching itself.
Unfortunately our techniques for maintaining the vertex set seem to be difficult to extend to the approximate setting, because they rely on distinguishing size $n$ vs $n-1$ (i.e., existence of a perfect matching).

\subsection{Preliminaries/Notation}

We write $[n]$ for set $\{1,...,n\}$.
For matrix $\mM \in \F^{n\times n}$ and sets $A,B\subset[n]$, we let $\mM_{A,B}$ be the submatrix consisting of the rows and columns with index in $A$ and $B$ respectively.
Similarly, we let $\mM_{-A,-B}$ be the matrix obtained via deleting the rows and columns with index in $A$ and $B$ respectively.
For bipartite graphs $G=(U\cup W, E)$ with $|U|=|W|=n$, we define the biadjacency matrix $\mA \in \F^{U\times W}$ via $\mA_{u,w} = 1$ if $\{u,w\} \in E$ and $\mA_{u,w} = 0$ otherwise.
For permutation $\sigma : [n]\to [n]$ the sign of $\sigma$ is $\sign(\sigma)=(-1)^t$ where $t$ is the number of inversions (i.e., number of pairs $i<j$ with $\sigma(i)>\sigma(j)$).

\begin{lemma}[Schwartz-Zippel-Lemma]\label{lem:schwartzzippel}
    Let $\F$ be a field and $f$ be a non-zero polynomial of degree $d$ with variables $x_1,...,x_n$, i.e., $f \in \F[x_1,...,x_n]$.
    Let $r_1,...,r_n$ be i.i.d.~uniformly sampled elements of $\F$. Then $\P[f(r_1,...,r_n) = 0] \le d/|\F|$.
\end{lemma}

Intuitively, this lemma states that if there exists some input for $f$ which evaluates to non-zero, then finding such an input is trivial by simply sampling a random input.

\begin{lemma}[{\cite[Lemma 2.4\&2.5]{CheungKL13}}]\label{lem:rankprojection}
    Given $n,k\in\N$, in $O(n)$ time we can construct a random $n\times r$ matrix $\mM\in\F^{n\times r}$ with $r=\Theta(k)$, such that

    \begin{itemize}
    \item each column of $\mM$ has at most $2\lceil n/r \rceil$ non-zero entries, and each row has 2 non-zero entries.

    \item for any matrix $\mA \in \F^{m\times n}$, we have with probability at least $1-O(1/k+k/|\F|)$ that
    $\rank(\mA\mM) = \min(\rank(\mA), r)$.
    \end{itemize}
\end{lemma}

The matrix $\mM$ in \Cref{lem:rankprojection} can be seen as a rank-preserving random projection. It projects $m$ many $n$-dimensional vectors (the rows of $\mA$) onto a smaller $r$-dimensional space such that with constant probability the rank is preserved.
Since the new space is $r$-dimensional, the rank is preserved only when at most $r$, hence $\rank(\mA\mM) = \min (\rank(\mA), r)$.
Further, the projection $\mM$ is sparse, i.e., changing an entry of a row-vector $v^\top$ will change only $2$ entries in its projection $v^\top \mM$.

\subparagraph{Remark:}
We can assume without loss of generality that field $\F$ has size $\poly(n)$, otherwise we use a field extension of such size \cite[Lemma 2.1]{CheungKL13}. 
Further, there is $k_0 = \Theta(1)$ such that for all $k>k_0$, the success probability of \Cref{lem:rankprojection} is at least 1/2.
We can then further boost the success probability of \Cref{lem:rankprojection} by sampling several $\mM_1,...,\mM_\ell$ and returning $\max_{i\in[\ell]}\rank(\mA\mM_i)$.
Then we have w.h.p., that the returned rank is $\min(\rank(\mA),r)$.

\section{Technical Overview}

We begin by outlining how to maintain a maximum sized full-rank submatrix (\Cref{thm:main_submatrix}) in \Cref{sec:overview:submatrix}.
The idea of the algorithm is based on constructing a gadget matrix by interpreting the full-rank submatrix problem graphically using connections between the determinant of a matrix and perfect matchings in graphs. 
Rather than a full-rank submatrix, we identify a vertex-set that must be deleted such that a perfect matching becomes possible. 

Then, in \Cref{sec:overview:rankbasis}, we explain how to improve the maintenance of rank (\Cref{thm:main_rank}) and basis (\Cref{thm:main_basis}) for low-rank matrices. These latter results are based on a random projection technique by \cite{CheungKL13}.

\subsection{Maintaining Full-Rank Submatrix (\Cref{sec:submatrix})}
\label{sec:overview:submatrix}

The task of finding a maximum sized full-rank submatrix for $n\times n$ matrix $\mA$ is equivalent to finding a minimum set $S,T \subseteq [n]$ of rows and columns to delete from $\mA$ such that $\det(\mA_{-S,-T}) \neq 0$.
Identifying these sets will rely on the following structural \Cref{lem:blockmatrixdeterminant}, proven in \Cref{sec:submatrix}.
In \Cref{lem:blockmatrixdeterminant}, matrix $\mA$ is the original input matrix and $\mB$ will be some gadget matrix that we construct.

\begin{restatable*}{lemma}{blockMatrixLemma}\label{lem:blockmatrixdeterminant}
    Suppose $\mA\in R^{n\times n}$ and $\mB\in R^{m\times m}$ with $m\ge n$ for some commutative ring $R$. Let
    \begin{align*}
        \mH=\left[\begin{array}{c|c}
            \mA&
            \begin{array}{ccc|c}
                x_1&&&0\hdots0\\
                &\ddots&&\vdots\\
                &&x_n&0\hdots0
            \end{array}
            \\
            \hline
            \begin{array}{ccc}
                y_1\\
                &\ddots\\
                &&y_n\\
                \hline
                \vspace{-5pt}0 &\hdots& 0\\
                \vspace{-7pt}&\vdots&\\
                0 &\hdots& 0\\
            \end{array}
            &
            \mB
        \end{array}\right].
    \end{align*}
    Then,
    \begin{align*}
        \det(\mH)=\sum_{S,T\subseteq[n],|S|=|T|}(-1)^{|S|}\det(\mA_{-S,-T})\det(\mB_{-T,-S})x^Sy^T.
    \end{align*}
    
In particular, 
$\det(\mH)\neq 0$ as a polynomial in $x_1,\ldots,x_n,y_1,\ldots,y_n$ iff there exists $S,T\subseteq [n]$ with $\det(\mA_{-S,-T})\neq 0$ and $\det(\mB_{-T,-S})\ne 0$.

\end{restatable*}

Via the Schwartz-Zippel-\Cref{lem:schwartzzippel} and replacing $x_1,...,x_n,y_1,...,y_n$ by random field elements, \Cref{lem:blockmatrixdeterminant} can be restated as follows:

\begin{fact}\label{fact:overview:matrix}
    W.h.p.~matrix $\mH$ has $\det(\mH) \neq 0$
    if and only if

    there is $S,T\subset[n]$
    where $\det(\mA_{-S,-T}) \neq 0$
    and $\det(\mB_{-T,-S}) \neq 0$.
\end{fact}

We maintain whether $\det(\mH)\neq 0$, for some gadget matrix $\mB$. We construct $\mB$ in such a way so that we know the collection of $S,T\subseteq[n]$ with $\det(\mB_{-T,-S}) \neq 0$.
Thus, when $\det(\mH)\neq 0$, we know that $\det(\mA_{-S,-T})\ne 0$ for at least one of these $S,T$.
In particular, by maintaining $\det(\mH)$ for gadget matrix $\mB$, we can implement a binary search procedure to determine a minimum $S,T$ with $\det(\mA_{-S,-T})\neq 0$.

\subparagraph{Invariant and Update Procedure}
Our dynamic algorithm will maintain the following invariant after each update to $\mA$. We maintain sets $S,T \subseteq[n]$ and a matrix $\mB$ such that
\ifdefined\arxiv
\begin{enumerate}[label=(\roman*)]
\else
\begin{enumerate}[(i)]
\fi
    \item $S,T$ are the unique subsets of $[n]$ for which $\det(\mB_{-T,-S})\neq 0$. Note that $\mB$ is $m\times m$ with $m\ge n$ and $S,T\subset[n]$, so only the first $n$ rows/columns are allowed to be deleted.
    \label{prop:overview:matching}
    \item $\mA_{-S,-T}$ is a maximum sized full-rank submatrix of $\mA$. \label{prop:overview:rank}
\end{enumerate}
By property 
\ifdefined\arxiv
\ref{prop:overview:rank},
\else
\eqref{prop:overview:rank},
\fi
we know a full-rank submatrix of $\mA$, which is exactly what our dynamic algorithm is supposed to return after each update.
Property 
\ifdefined\arxiv
\ref{prop:overview:matching}
\else 
\eqref{prop:overview:matching}
\fi
will hold after we are done processing the update to $\mA$, but uniqueness does not hold while processing the update, as we are performing a binary search to find one valid $S,T$.

To outline how we maintain these two properties, consider the following possibilities after we receive an entry-update to $\mA$.
Since a single entry $\mA_{i,j}$ is updated, the rank of $\mA$ can change by at most 1, giving us the following cases:
\ifdefined\arxiv
\begin{enumerate}[label=(\alph*)]
\else
\begin{enumerate}[(1)]
\fi
\item The rank increases, so we can add one more row and column to the full-rank submatrix of $\mA$. In particular, our task is to find and remove an index from $S$ and an index from $T$, such that 
\ifdefined\arxiv
\ref{prop:overview:rank}
\else 
\eqref{prop:overview:rank}
\fi
holds again. \label{case:overview:a}
\item The rank decreases, so we must remove a row and column from the (previously) full-rank submatrix of $\mA$. In particular, the task is to find and add an index to $S$ and an index to $T$, such that 
\ifdefined\arxiv
\ref{prop:overview:rank}
\else 
\eqref{prop:overview:rank}
\fi
holds again. This task is trivial since updating entry $\mA_{i,j}$ means that the $i$-th row and $j$-th column are precisely the ones to be removed from the submatrix, i.e., $S\gets S\cup\{i\}, T\gets T\cup \{j\}$. \label{case:overview:b}
\item The rank stays the same. We can reduce this to case 
\ifdefined\arxiv
\ref{case:overview:b}+\ref{case:overview:a}:
\else 
\eqref{case:overview:b}+\eqref{case:overview:a}:
\fi
After an update to $\mA_{i,j}$, we add $i,j$ to $S,T$, then attempt to remove an index from $S$ and from $T$ again via the procedure of 
\ifdefined\arxiv
\ref{case:overview:a}.
\else 
\eqref{case:overview:a}.
\fi
\end{enumerate}

So our update routine is as follows:
Before updating entry $\mA_{i,j}$, we check if the update will make $\mA_{-S,-T}$, and thus $\mH$, singular. If it would, we add $i$ to $S$ and $j$ to $T$ before making the update. In either case, we update $\mA_{i,j}$ and then attempt to find an index to remove from $S$ and an index to remove from $T$ and remove them. We do this last step (searching for and removing a pair of indices) twice.

We are left with describing the procedure for finding $s \in S$ and $t \in T$ such that after removal we still have $\det(\mA_{-S,-T}) \neq 0$.
To do this, we construct our gadget matrix $\mB$ such that the following holds:

\begin{property}\label{prop:overview:matrix}
Given $S,T \subset [n]$ and $C=\{l,...,r\}, C'=\{l',...,r'\}$, the matrix $\mB$ has the property that all $S',T' \subset [n]$ with $\det(\mB_{-T',-S'})\neq 0$ are precisely the sets $S' = S \setminus \{c\}, T'=T\setminus \{c'\}$ for $c \in C \cap S, c'\in C' \cap T$.
\end{property}

Thus, when $\det(\mH) \neq 0$, then know there is some $S',T'$ with $\det(\mA_{-S',-T'})\neq 0$, and by repeatedly halving the size of $C,C'$ we can binary search for a single valid $S',T'$.

Given the binary search structure, it suffices to restrict to $C,C'$ to be certain intervals. (When $n$ is a power of two, they will have the form $\{k\cdot 2^\ell+1,..., (k+1) 2^\ell \}$ for $\ell,k \in \N$).
We show that for $S,T$ and all the $C,C'$ needed in the binary search, there is a matrix $\mB$ with \Cref{prop:overview:matrix}. Further, changing one entry in $S,T$ and halving sets $C,C'$ results in only $O(1)$ changed entries in $\mB$.
Thus, after an update to matrix $\mA$, we only perform $O(\log n)$ updates to $\mB$ and thus only $O(\log n)$ updates to $\mH$.
In particular, when given a dynamic algorithm maintaining $\det(\mH)\stackrel{?}{=}0$ with update time $U$, we have a dynamic algorithm maintaining a full-rank submatrix of $\mA$ in update time $O(U \log n)$ (\Cref{thm:main_submatrix}\footnote{\Cref{thm:main_submatrix} states that it suffices to detect when $\det(\mH)=0$ for the first time. That's because by our construction of $\mB$, any updating causing $\det(\mH)=0$ will be followed by an update that reverts the previous one. So one can simply revert and skip the previous update that made the matrix singular.}).

\subparagraph{Gadget Matrix $\mB$ via Gadget Forest $F$.}

We get to choose matrix $\mB$ in our algorithm.
To have precise control about the possible sets $S,T$ where $\det(\mB_{-T,-S})\neq 0$ we will use a graph theoretic perspective based on perfect matchings in bipartite graphs. For any $m\times m$ matrix $\mB$ we can write the determinant as
$$
\det(\mB) = \sum_{\sigma\in\cS_m} \sign(\sigma) \prod_{i=1}^m \mB_{i,\sigma(i))}.
$$
Here $\sigma$ are permutations over $[m]$.
When $\mB$ represents a bipartite graph, i.e., edge $\{i,j\}$ exists if and only if $\mB_{i,j}\neq 0$, then each non-zero product in above sum represents a valid perfect matching (given by $(i,\sigma(i))$) in the graph.
The main issue is that $\sign(\sigma)\in\{\pm 1\}$ could lead to cancellations if the perfect matching is not unique. To enforce uniqueness, we let the graph represented by $\mB$ be a forest.

We let $\mB\in\F^{m\times m}$ be the biadjacency matrix of a forest $F=((U,V),E)$ on $2m$ vertices. Here $U=[m],V=[m]$ are the left/right vertices since a forest is bipartite. In particular, the row indices of $\mB$ correspond to $U$ and the column indices correspond to $V$. 
As such, $B_{-T,-S}$ is also forest since one only deleted the vertices $S \subset U$ and $T \subset V$ from forest $F$.
Thus, the condition $\det(\mB_{-T,-S})\neq 0$ is equivalent to $F[U\setminus S, V \setminus T]$ (i.e., graph $F$ after deleting vertices $S$ and $T$) having a perfect matching.

Hence from now on, rather than argue about sets $S,T$ with $\det(\mB_{-T,-S})\neq 0$, we will argue about vertex sets whose removal allows forest $F$ to have a perfect matching.

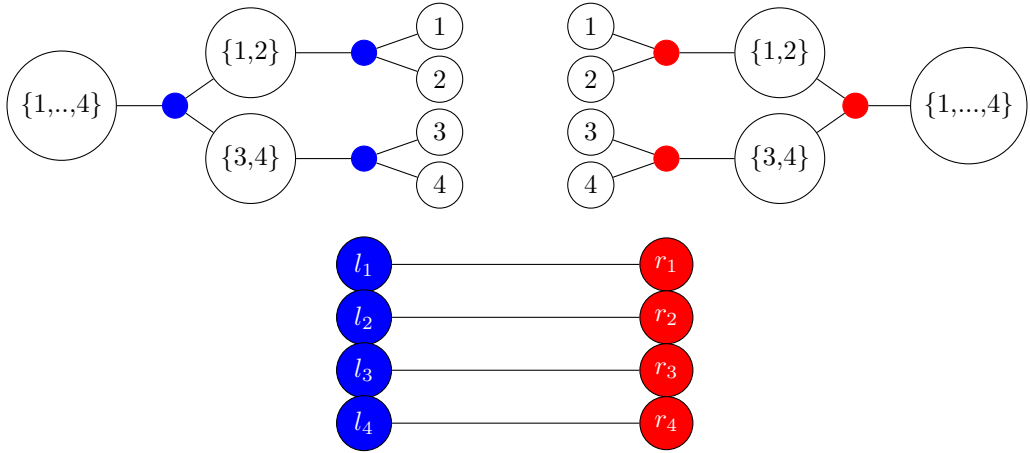
\begin{figure}[th!]
\centering
\begin{tikzpicture}[yscale=0.7]

\node[circle,draw] (LUUA) at (-2.0,1.5) {1};
\node[circle,draw] (LUDA) at (-2.0,0.5) {2};
\node[circle,draw] (LUA) at (-4.5,1.0) {\{1,2\}};
\node[circle,draw,color=blue,fill] (LUB) at (-3.0,1.0) {};
\node[circle,draw] (LDUA) at (-2.0,-0.5) {3};
\node[circle,draw] (LDDA) at (-2.0,-1.5) {4};
\node[circle,draw] (LDA) at (-4.5,-1.0) {\{3,4\}};
\node[circle,draw,color=blue,fill] (LDB) at (-3.0,-1.0) {};
\node[circle,draw] (LA) at (-7.0,0.0) {\{1,..,4\}};
\node[circle,draw,color=blue,fill] (LB) at (-5.5,0.0) {};
\path[circle,draw] (LA) edge (LB);
\path[circle,draw] (LUA) edge (LB);
\path[circle,draw] (LDA) edge (LB);
\path[circle,draw] (LUA) edge (LUB);
\path[circle,draw] (LUUA) edge (LUB);
\path[circle,draw] (LUDA) edge (LUB);
\path[circle,draw] (LDA) edge (LDB);
\path[circle,draw] (LDUA) edge (LDB);
\path[circle,draw] (LDDA) edge (LDB);

\node[circle,draw] (RUUA) at (0.0,1.5) {{1}};
\node[circle,draw] (RUDA) at (0.0,0.5) {{2}};
\node[circle,draw] (RUA) at (2.5,1.0) {\{1,2\}};
\node[circle,draw,color=red,fill] (RUB) at (1.0,1.0) {};
\node[circle,draw] (RDUA) at (0.0,-0.5) {{3}};
\node[circle,draw] (RDDA) at (0.0,-1.5) {{4}};
\node[circle,draw] (RDA) at (2.5,-1.0) {\{3,4\}};
\node[circle,draw,color=red,fill] (RDB) at (1.0,-1.0) {};
\node[circle,draw] (RA) at (5.0,0.0) {\{1,...,4\}};
\node[circle,draw,color=red,fill] (RB) at (3.5,0.0) {};
\path[circle,draw] (RA) edge (RB);
\path[circle,draw] (RUA) edge (RB);
\path[circle,draw] (RDA) edge (RB);
\path[circle,draw] (RUA) edge (RUB);
\path[circle,draw] (RUUA) edge (RUB);
\path[circle,draw] (RUDA) edge (RUB);
\path[circle,draw] (RDA) edge (RDB);
\path[circle,draw] (RDUA) edge (RDB);
\path[circle,draw] (RDDA) edge (RDB);

\node[circle,draw,fill=blue,text=white] (L1) at (-3,-3) {$l_1$};
\node[circle,draw,fill=red,text=white] (R1) at (1,-3) {$r_1$};
\node[circle,draw,fill=blue,text=white] (L2) at (-3,-4) {$l_2$};
\node[circle,draw,fill=red,text=white] (R2) at (1,-4) {$r_2$};
\node[circle,draw,fill=blue,text=white] (L3) at (-3,-5) {$l_3$};
\node[circle,draw,fill=red,text=white] (R3) at (1,-5) {$r_3$};
\node[circle,draw,fill=blue,text=white] (L4) at (-3,-6) {$l_4$};
\node[circle,draw,fill=red,text=white] (R4) at (1,-6) {$r_4$};

\path[circle,draw] (L1) edge (R1);
\path[circle,draw] (L2) edge (R2);
\path[circle,draw] (L3) edge (R3);
\path[circle,draw] (L4) edge (R4);
\end{tikzpicture}
\caption{\label{fig:overview:blankforest}
Base forest construction, before adding any dependence on $S,T,C,C'\subset [n]$.}
\end{figure}

\subparagraph{Gadget Forest}

\textit{Base Gadget.}~~
We first describe a forest $F$ that does not yet depend on given sets $S,T,C,C'$. We will then modify the forest for some given $S,T,C,C'$.
The forest is given by the following construction, depicted in \Cref{fig:overview:blankforest}.
We have two binary trees with $n$ leaves where every internal vertex is split into two (a colored and a white vertex), with the colored vertex having the children and the white vertex having the parent.
The leaves of the two trees are labeled with $1,...,n$. 
In particular, sets $S,T\subset[n]$ of vertices that are to be deleted from $F$ will correspond to the left/right leaves.
Any white vertex is labeled by the set of leaves in that subtree, thus any possible set\footnote{Recall that $C,C'$ are not arbitrary sets but facilitate a binary search. Thus they match the vertex labels of the binary tree.}$C,C'$ corresponds to one white vertex of the tree on the left/right.

Additionally, we have $2n$ vertices ($\ell_1,...,\ell_n$ on the left and $r_1,...,r_n$ on the right respectively) connected to their horizontal neighbor, i.e., $l_i$ connects to $r_i$.
We refer to these vertices also as the ``switch-vertices'' and they will later be connected to other vertices in the forest based on some given sets $S,T,C,C'$.\medskip

\textit{Modification based on $S$ and $T$.}~~
Recall that $S,T$ must be the (unique) set of leaves to be deleted so the forest has a perfect matching.
If we want the forest in \Cref{fig:overview:blankforest} to have a perfect matching, then each white tree node must be matched to its colored child. 
So the only way to have a perfect matching is if we were to delete all leaves. Since we want to have to delete only $S,T$, we do the following: we disconnect pairs of switch-vertices and connect them to leaves not in $S,T$. Then those leaves must be matched to the respective switch-vertices and must no longer be deleted. An example is given in \Cref{fig:overview:STforest}.

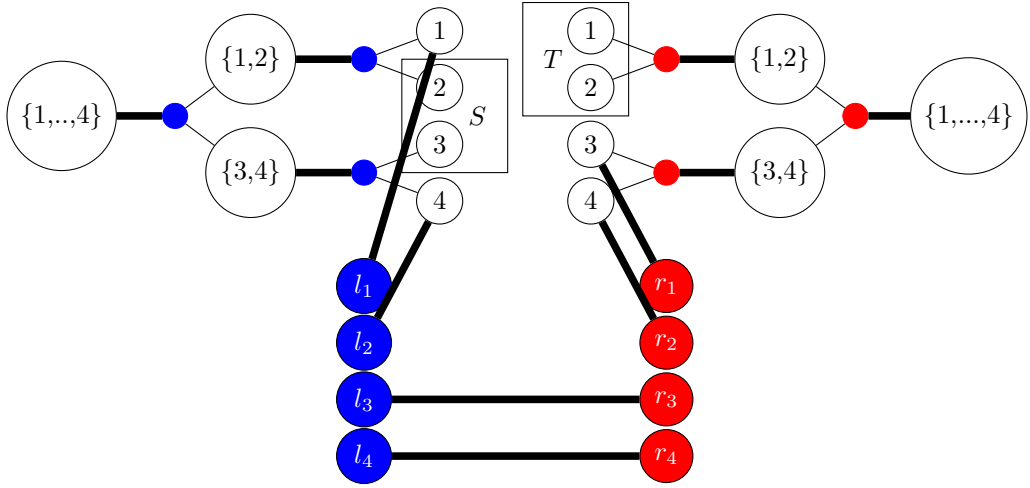
\begin{figure}[th]
\centering
\begin{tikzpicture}[yscale=0.75]

\node (S) at (-1.5,0) {$S$};
\draw (-2.5,1) -- (-1.1,1) -- (-1.1,-1) -- (-2.5,-1) -- cycle;

\node (T) at (-0.5,1) {$T$};
\draw (0.5,2) -- (-0.9,2) -- (-0.9,0) -- (0.5,0) -- cycle;

\node[circle,draw] (LUUA) at (-2.0,1.5) {1};
\node[circle,draw] (LUDA) at (-2.0,0.5) {2};
\node[circle,draw] (LUA) at (-4.5,1.0) {\{1,2\}};
\node[circle,draw,color=blue,fill] (LUB) at (-3.0,1.0) {};
\node[circle,draw] (LDUA) at (-2.0,-0.5) {3};
\node[circle,draw] (LDDA) at (-2.0,-1.5) {4};
\node[circle,draw] (LDA) at (-4.5,-1.0) {\{3,4\}};
\node[circle,draw,color=blue,fill] (LDB) at (-3.0,-1.0) {};
\node[circle,draw] (LA) at (-7.0,0.0) {\{1,..,4\}};
\node[circle,draw,color=blue,fill] (LB) at (-5.5,0.0) {};
\path[line width=1mm] (LA) edge (LB);
\path[circle,draw] (LUA) edge (LB);
\path[circle,draw] (LDA) edge (LB);
\path[line width=1mm] (LUA) edge (LUB);
\path[circle,draw] (LUUA) edge (LUB);
\path[circle,draw] (LUDA) edge (LUB);
\path[line width=1mm] (LDA) edge (LDB);
\path[circle,draw] (LDUA) edge (LDB);
\path[circle,draw] (LDDA) edge (LDB);

\node[circle,draw] (RUUA) at (0.0,1.5) {{1}};
\node[circle,draw] (RUDA) at (0.0,0.5) {{2}};
\node[circle,draw] (RUA) at (2.5,1.0) {\{1,2\}};
\node[circle,draw,color=red,fill] (RUB) at (1.0,1.0) {};
\node[circle,draw] (RDUA) at (0.0,-0.5) {{3}};
\node[circle,draw] (RDDA) at (0.0,-1.5) {{4}};
\node[circle,draw] (RDA) at (2.5,-1.0) {\{3,4\}};
\node[circle,draw,color=red,fill] (RDB) at (1.0,-1.0) {};
\node[circle,draw] (RA) at (5.0,0.0) {\{1,...,4\}};
\node[circle,draw,color=red,fill] (RB) at (3.5,0.0) {};
\path[line width=1mm] (RA) edge (RB);
\path[circle,draw] (RUA) edge (RB);
\path[circle,draw] (RDA) edge (RB);
\path[line width=1mm] (RUA) edge (RUB);
\path[circle,draw] (RUUA) edge (RUB);
\path[circle,draw] (RUDA) edge (RUB);
\path[line width=1mm] (RDA) edge (RDB);
\path[circle,draw] (RDUA) edge (RDB);
\path[circle,draw] (RDDA) edge (RDB);

\node[circle,draw,fill=blue,text=white] (L1) at (-3,-3) {$l_1$};
\node[circle,draw,fill=red,text=white] (R1) at (1,-3) {$r_1$};
\node[circle,draw,fill=blue,text=white] (L2) at (-3,-4) {$l_2$};
\node[circle,draw,fill=red,text=white] (R2) at (1,-4) {$r_2$};
\node[circle,draw,fill=blue,text=white] (L3) at (-3,-5) {$l_3$};
\node[circle,draw,fill=red,text=white] (R3) at (1,-5) {$r_3$};
\node[circle,draw,fill=blue,text=white] (L4) at (-3,-6) {$l_4$};
\node[circle,draw,fill=red,text=white] (R4) at (1,-6) {$r_4$};

\path[line width=1mm] (L1) edge (LUUA);
\path[line width=1mm] (R1) edge (RDUA);
\path[line width=1mm] (L2) edge (LDDA);
\path[line width=1mm] (R2) edge (RDDA);
\path[line width=1mm] (L3) edge (R3);
\path[line width=1mm] (L4) edge (R4);
\end{tikzpicture}
\caption{\label{fig:overview:STforest}For $S=\{2,3\}$ the left leaves $1,4$ are connected to switch-vertices. For $T=\{1,2\}$, the right leaves $3,4$ are connected to switch-vertices. Set $S,T$ are the unique set of leaves after whose deletion the forest has a perfect matching (bold lines).}
\end{figure}

\textit{Modification based on $C$ and $C'$.}~~
Given $C$ (and $C'$), we disconnect another pair of switch-vertices and connect them to the white vertices labeled with $C$ (and $C'$) of the left (and right) tree.
Observe that, since a switch-vertex always has a single neighbor, it must be matched to that neighbor. So connecting a switch-vertex to some vertex $v$ is equivalent to deleting both the switch-vertex and $v$, when it comes to the graph having a perfect matching. So to simplify \Cref{fig:overvoew:CCtree} we removed all vertices connected to switch-vertices. 

\begin{figure}[th]
\centering
\begin{tikzpicture}[yscale=0.75]

\node (S) at (-1.5,0) {$S$};
\draw (-2.5,1) -- (-1,1) -- (-1,-1) -- (-2.5,-1) -- cycle;

\node (T) at (1.5,1) {$T$};
\draw (2.5,2) -- (1,2) -- (1,0) -- (2.5,0) -- cycle;

\node[circle,draw] (LUDA) at (-2.0,0.5) {2};
\node[circle,draw] (LUA) at (-4.5,1.0) {\{1,2\}};
\node[circle,draw,color=blue,fill] (LUB) at (-3.0,1.0) {};
\node[circle,draw] (LDUA) at (-2.0,-0.5) {3};
\node[circle,draw] (LDA) at (-4.5,-1.0) {\{3,4\}};
\node[circle,draw,color=blue,fill] (LDB) at (-3.0,-1.0) {};
\node[circle,draw,color=blue,fill] (LB) at (-5.5,0.0) {};
\path[circle,draw] (LUA) edge (LB);
\path[circle,draw] (LDA) edge (LB);
\path[line width=1mm] (LUA) edge (LUB);
\path[circle,draw] (LUDA) edge (LUB);
\path[line width=1mm] (LDA) edge (LDB);
\path[circle,draw] (LDUA) edge (LDB);

\node[circle,draw] (RUUA) at (2.0,1.5) {{1}};
\node[circle,draw] (RUDA) at (2.0,0.5) {{2}};
\node[circle,draw,color=red,fill] (RUB) at (3.0,1.0) {};
\node[circle,draw] (RDA) at (4.5,-1.0) {\{3,4\}};
\node[circle,draw,color=red,fill] (RDB) at (3.0,-1.0) {};
\node[circle,draw] (RA) at (7.0,0.0) {\{1,...,4\}};
\node[circle,draw,color=red,fill] (RB) at (5.5,0.0) {};
\path[line width=1mm] (RA) edge (RB);
\path[circle,draw] (RDA) edge (RB);
\path[circle,draw] (RUUA) edge (RUB);
\path[circle,draw] (RUDA) edge (RUB);
\path[line width=1mm] (RDA) edge (RDB);

\end{tikzpicture}
\caption{\label{fig:overvoew:CCtree}$S=\{1,2\},T=\{1,2\},C=\{1,2,3,4\},C'=\{1,2\}$, vertices connected to switch-vertices have been removed for simplicity. Observe that $S' = S \setminus \{c\}$ for $c \in C \cap S$ are the only sets of leaves whose deletion allows the blue forest to have a perfect matching. The choice of $c$ corresponds to an augmenting path from the unmatched blue vertex to $c$. Likewise $T' = T \setminus \{c'\}$ for $c' \in C' \cap T$ are the possible leaf deletions such that the red forest has a perfect matching.}
\end{figure}
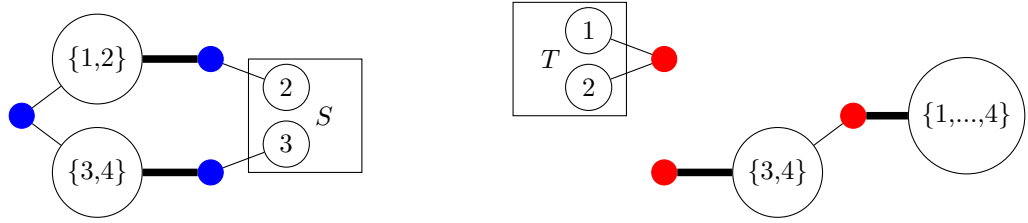

Since internal tree vertex $C$ (e.g., a set like $\{1,2,3,4\}$) is removed, it can no longer be matched to its colored child, so the colored child must be matched to either the top or bottom half of set $C$ (e.g., $\{1,2\}$ or $\{3,4\}$).
We essentially get an augmenting path along the tree from the removed vertex $C$ to one of the leaves. This only works if the leaf was not removed, i.e., not connected to a switch-vertex, but also the leaf must not be in $S$ (since $S$ is the set of leaves we want to delete).
Thus one can show that, only deletion of leaves $S'=S\setminus\{c\}$ for any $c \in C \cap S$ allows the tree to have a perfect matching.
The same argument holds for the other tree, deleting leaves $T' = T \setminus \{c'\}$ for any $c'\in C'$ allows the tree to have a perfect matching.

Thus, given $S,T,C,C'$ we have a forest where the possible set of leaf deletions such that the forest has a perfect matching, are all of the form $S' = S \setminus\{c\}$, $T' = T \setminus \{c'\}$ for any $c\in C\cap S$, $c' \in C' \cap T$.
Using the bi-adjacency matrix of this forest thus yields a matrix $\mB$ with \Cref{prop:overview:matrix} required for our algorithm.
Further, changing an entry in $S$ or $T$, or replacing $C$ or $C'$ each will change only $O(1)$ edges in the forest (since we only reconnect two switch-vertices), i.e., $O(1)$ entries in matrix $\mB$.

\ifdefined\arxiv
\subsection{Maintaining Small Rank\&Basis (\Cref{sec:smallrank})}
\else
\subsection{Maintaining Small Rank\&Basis}
\fi
\label{sec:overview:rankbasis}

The complexity of the submatrix-algorithm outlined in previous subsection scales in dimension $n$.
Here we outline how to maintain small rank more efficiently (\Cref{thm:main_rank}), and how to maintain a basis (\Cref{thm:main_basis}) more efficiently for low-rank matrices.
Our results rely on \Cref{lem:rankprojection} (see preliminaries) by \cite{CheungKL13}.

Maintaining just the rank is simple with the above lemma, and we start by outlining that algorithm.
Then we extend the algorithm to also maintain a basis.

\subparagraph{Maintaining Rank (\Cref{thm:main_rank,thm:rankreduction})}

Let us assume for now that we just want to maintain $\rank(\mA)$ up to some parameter $k$, i.e., we maintain $\min(\rank(\mA),k)$.

To do this, we first sample two matrices $\mL,\mR$ via \Cref{lem:rankprojection} and then 
compute $\mL^\top \mA \mR$.
By \Cref{lem:rankprojection} we have $\rank(\mL^\top \mA \mR) = \min(k, \rank(\mA\mR)) = \min(k, \rank(\mA))$.
So by maintaining the rank of $\mL^\top \mA \mR$, we maintain the rank of $\mA$ up to $k$.

Further, by \Cref{lem:rankprojection}, there are only 2 non-zero entries per row in $\mL,\mR$.
Thus an entry-update to $\mA$ changes only 4 entries of $\mL^\top \mA \mR$. Also observe that matrix $\mL^\top \mA \mR$ is of size $ O(k)\times O(k)$ due to width of $\mL,\mR$.
Thus if we can maintain the rank of an $n\times n$ matrix in $U(n)$ update time, then we can now maintain $\min(\rank(\mA), k)$ in $4\cdot (U( O(k))) = O (U(k))$ update time. We run $O(\log n)$ independent copies of this so at least one preserved the rank with high probability. Thus resulting in $\tilde O(U(k))$ update time.

Next, we extend this to $\tilde O(\rank(\mA))$ update time.
To do this, we maintain the product $\mL^\top \mA \mR$ for different powers of $k=2^i$, $i=0,...,\log n$.
Whenever the rank of $\mA$ changes by constant factor, we initialize a new dynamic rank data structure for some $k=2^i$ with $i=\lceil\log(\rank(\mA))\rceil+1$.
Until we must initialize a data structure for $k=2^{i\pm 1}$, there must be at least $2^{i-1}$ updates. (Because each update can change the rank by at most 1.)
So the initialization cost amortizes over $O(2^i)$ updates.
This leads to
an amortized update time of
$$\tilde O( U(\rank(\mA)) + P(\rank(\mA))/\rank(\mA)),$$
where $P(n),U(n)$ are the preprocessing and update time complexity of a dynamic rank data structure for $n\times n$ matrices.

In case of column-updates, the same techniques works, because changing one column of $\mA$ changes at most 2 columns of $\mL^\top\mA \mR$. So there, too, the reduction has only a constant-factor blow-up for the number of updates.

\subparagraph{Maintaining Basis (\Cref{thm:main_basis,thm:basis_reduction})}

When performing a column update, the rank can change by at most 1.
In particular, if the updated column was already part of the basis then it might have to be removed or replaced.
If the updated column was not part of the basis, then it might need to be added to the basis.
In either case, the task of maintaining a basis reduces to the following problem:
Given dynamic matrix $\mA$ and a linearly-independent set of column-vectors $\mB$, find if there is a column in $\mA$ that can be added to $\mB$.

Observe that if all columns of $\mA$ are linearly dependent to $\mB$, then any linear combination $\mA v$ is also dependent to $\mB$.
Conversely, if $\mA$ contains a column that is linearly independent, then there exists $v=e_i$ (a standard unit vector) such that $\mA v$ is independent of $\mB$.
By Schwartz-Zippel-\Cref{lem:schwartzzippel}, this then implies\footnote{Consider for instance the polynomial $\det([\mB|\mA v]^\top [\mB| \mA v])$ which is non-zero if and only if $[\mB| \mA v]$ (the matrix $\mB$ with new column $\mA v$ appended) is linearly independent.} that for random vector $v$, we also have that $\mA v$ is independent of $\mB$ with high probability.
Now assume we have a dynamic rank data structure supporting column updates. If adding column $\mA v$ to $\mB$ increases the rank, we know $\mA$ contains an independent column.
We can use this to binary-search for the first independent column in $\mA$, by repeatedly zeroing out half the entries of $v$.
Each step in the binary-search corresponds to adding and removing a column $\mA v$ to $\mB$.
So we need $O(\log n)$ column-updates to find a new linearly independent column.

The last bottleneck is that computing $\mA v$ takes $O(n^2)$ time.
Thus we must maintain the possible $\mA v$-products that occur during the binary-search, rather than computing them from scratch in each iteration.

During initialization, we precompute all $\mA v$-products that could show up during binary-search.
Then when a column of $\mA$ changes, there are only $O(\log n)$ matrix-vector-products that must be updated. In particular, updating one product $\mA v$ during a column-update to $\mA$ takes only $O(z \log n)$ time, where $z$ is the number of entries changed in $\mA$ by the update.

In summary, let $U(n)$ be the update time of maintaining rank of a dynamic $n\times n$ matrix $\mH$, subject to column-updates.
Then we can maintain a column basis in time
$$O((U(n) + z) \log n).$$
This is $\tilde O(n^{\columnUpdate}+z)$ update time using \cite{BrandNS19}.
Since the reduction only relies on detecting whether $\rank(\mB)$ increases or decreases, we can run the rank data structure on $\mM^\top \mB$ for random matrix $\mM$ from \Cref{lem:rankprojection}, improving the complexity dependence in $n$ to $\rank(\mA)$. Similarly, instead of maintaining the products $\mA v$ for the binary-search, we maintain $(\mM^\top \mA) v$. Since an entry update to $\mA$ changes only 2 entries of $O(\mM^\top \mA )$, there is no complexity increase.
This leads to \Cref{thm:main_basis}, i.e., $\tilde O(\rank(\mA)^{\columnUpdate}+z)$ update time.

\section{Maintain Submatrices}
\label{sec:submatrix}

In this section, we prove our reduction from maximum full-rank submatrix to singularity detection.

\thmSubmatrix*

As outlined in \Cref{sec:overview:submatrix}, the reduction is simple if we assume all non-zero entries of the matrix are random field elements, because then full-rankness is equivalent to perfect bipartite matching.
The main challenge is to prove an equivalent result when the input matrix does not have random entries.
To do this, we translate the graph reduction from \Cref{sec:overview:submatrix} to matrix perspective. Hence we strongly recommend readers to first read \Cref{sec:overview:submatrix} before reading \Cref{sec:submatrix}.

The gadget-graph described in \Cref{sec:overview:submatrix}, is translated to a biadjacency matrix, whose determinant is analyzed in \Cref{sec:blockmatrixdeterminant}. That analysis requires careful handling of signs of permutations, which would not have been an issue in the random-entry case.
In particular, we relate the determinant of the original input matrix $\mA$ with the determinant of the biadjacency matrix that represents the binary trees in \Cref{sec:overview:submatrix}.

In \Cref{sec:submatrix:treegadget}, we analyze the determinant of the binary-tree-matrix and analyze the determinant behavior when performing the binary-search outlined in \Cref{sec:overview:submatrix}.

Lastly, in \Cref{sec:submatrix:combine} we combine the results and prove that the binary-search yields an algorithm for maintaining a maximum full-rank submatrix.

\subsection{Block Matrix Determinant}
\label{sec:blockmatrixdeterminant}

In this subsection, we prove the following lemma, which relates the determinants of matrix $\mA$ and $\mB$. Here matrix $\mA$ will be the original input matrix whose submatrix we want to maintain.
Matrix $\mB$ will correspond to the biadjacency matrix of the binary-tree gadgets in \Cref{sec:overview:submatrix}.
The lemma will be used in \Cref{sec:submatrix:combine}, so show that updates to $\mB$ can be used to binary-search for suitable rows and columns in $\mA$ that are linearly independent.

\blockMatrixLemma

\begin{proof}
Let $\mX\in R^{n\times m},\mY\in R^{m\times n}$ be defined by $\mX_{ii}=x_i$, $\mY_{ii}=y_i$ for $i\in[n]$ and $\mX_{ij}=\mY_{ji}=0$ for $i\ne j$.
\begin{align*}
    \det(\mH)=
    \det\begin{bmatrix}
        \mA&\mX\\
        \mY&\mB
    \end{bmatrix}&=
    \det\left(
    \begin{bmatrix}
        \mI&\mI\\&&\mI&\mI
    \end{bmatrix}
    \begin{bmatrix}
        \mA\\
        &&\mX\\
        &\mY\\
        &&&\mB
    \end{bmatrix}
    \begin{bmatrix}
        \mI\\\mI\\&\mI\\&\mI
    \end{bmatrix}
    \right)
\end{align*}
By the Cauchy–Binet formula we can write $\det(\mH)$ as
\begin{align*}
    \sum_{S,T\subseteq\binom{[2n+2m]}{n+m}}
    \det\begin{bmatrix}
        \mI&\mI\\&&\mI&\mI
    \end{bmatrix}_{[n+m],S}
    \det\begin{bmatrix}
        \mA\\
        &&\mX\\
        &\mY\\
        &&&\mB
    \end{bmatrix}_{S,T}
    \det\begin{bmatrix}
        \mI\\\mI\\&\mI\\&\mI
    \end{bmatrix}_{T,[n+m]}
\end{align*}

Note that the first (resp. last) determinant is zero unless $S$ (resp. $T$) selects exactly $n$ of the first $2n$ and $m$ of the last $2m$ columns (resp. rows). Thus, $\det(\mH)$ is equal
\begin{align*}
    \sum_{S,T\subseteq[n],U,V\subseteq[m]}
    (-1)^{\alpha(S)+\alpha(T)+\beta(U)+\beta(V)}
    \det\begin{bmatrix}
        \mA_{[n]\setminus S,[n]\setminus T}\\
        &&\mX_{S,V}\\
        &\mY_{U,T}\\
        &&&\mB_{[m]\setminus U,[m]\setminus V}
    \end{bmatrix}
\end{align*}
where we define
$$\alpha(S)=|\{(i,j):i\in S,j\in[n]\setminus S,i<j\}| \text{ and } \beta(U)=|\{(i,j):i\in[m]\setminus U,j\in U,i<j\}|.$$

The determinant is zero unless each of the four blocks is square, and $S=V$ and $T=U$.

\noindent
Note that for $S\subseteq[n]$, $\alpha(S)+\beta(S)=|S|(n-|S|)$. Thus, $(-1)^{\alpha(S)+\alpha(T)+\beta(U)+\beta(V)}=1$, and
\begin{align*}
    \det(\mH)&=\sum_{S,T\subseteq[n],|S|=|T|}
    (-1)^{|S|}
    \det(\mA_{[n]\setminus S,[n]\setminus T})
    \det(\mB_{[m]\setminus T,[m]\setminus S})
    x^Sy^T
\end{align*}

The $(-1)^{|S|}$ comes from swapping $\mX_{S,V}$ and $\mY_{U,T}$ onto the diagonal.
\end{proof}

\subsection{Binary-Search Gadget Matrix}
\label{sec:submatrix:treegadget}

\begin{figure}
    \centering
\begin{tikzpicture}[xscale=-0.15,yscale=0.15]
    \def\n{8}
    \draw[step=1cm, gray!10, very thin] (0,0) grid (\n*5-2,\n*5-2);
    \foreach \x in {0,\n,\n*2-1,\n*3-2,\n*4-2,\n*5-2} {
        \draw[gray] (\x,0) -- (\x,\n*5-2);
        \draw[gray] (0,\x) -- (\n*5-2,\x);
    }
    \foreach \i in {0,...,\the\numexpr\n-1\relax} {
        \fill[black] (\i,\i) rectangle +(1,1);
        \fill[black] (\i,\i) rectangle +(1,1);
        \fill[green] (\n*3-2+\i,\n*4-2+\i) rectangle +(1,1);
        \fill[cyan] (\n*4-2+\i,\n*3-2+\i) rectangle +(1,1);
    }
    \foreach \i in {0,...,\the\numexpr\n-2\relax} {
        \fill[black] (\n+\i,\n*2-1+\i) rectangle +(1,1);
        \fill[black] (\n+\i,\n*2-1+\i*2+1) rectangle +(1,1);
        \fill[black] (\n+\i,\n*2-1+\i*2+2) rectangle +(1,1);
        \fill[black] (\n*2-1+\i,\n+\i) rectangle +(1,1);
        \fill[black] (\n*2-1+\i*2+1,\n+\i) rectangle +(1,1);
        \fill[black] (\n*2-1+\i*2+2,\n+\i) rectangle +(1,1);
    }
    \node[scale=\n*0.2] at (\n*4.5-2,\n*4.5-2) {$A$};
\end{tikzpicture}
    \caption{Shape of matrix $\mH$ (\Cref{lem:blockmatrixdeterminant}), the green and cyan diagonals are the random $x_1,..,x_n$ and $y_1,...,y_n$. The black boxes represent the non-zero entries of matrix $\mB$, which is the biadjacency matrix of the forest given in \Cref{fig:overview:blankforest}. The bottom right diagonal are the edges connecting left and right ``switch-vertices.'' The black wedges represent the two binary trees.\label{fig:matrixH}}
\end{figure}
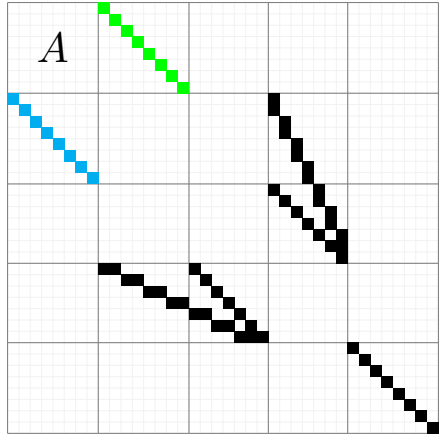

By choosing $\mB$, we can check if a certain subcollection of submatrices of $\mA$ all have determinant zero via \Cref{lem:blockmatrixdeterminant}.
Intuitively, matrix $\mB$ represents the biadjacency matrix of the forest in \Cref{sec:overview:submatrix}, e.g., \Cref{fig:overview:blankforest}. An example of what $\mH$ looks like for such $\mB$ is given in \Cref{fig:matrixH}.

We start by formally defining the graph gadget in this \Cref{sec:submatrix:treegadget}.
Then in \Cref{sec:submatrix:combine} we use those gadget to facilitate the binary-search for new rows/columns to be added to a full-rank submatrix.

First, we define the structure of the trees. \Cref{fig:submatrix:tree} shows their structure.
\begin{definition}\label{def:submatrix:treestructure}
For integer $r\ge l\ge 1$, let $\mathcal{T}_{l,r}$ be a rooted tree defined recursively as follows:
\begin{itemize}
    \item If $l=r$, then $\mathcal{T}_{l,r}$ consists only of a root labeled $\{l\}$
    \item If $l<r$, then $\mathcal{T}_{l,r}$ consists of a root labeled $\{l,\ldots,r\}$ with an (unlabeled) child vertex whose children are the roots of $\mathcal{T}_{l,\lfloor \frac{l+r}{2}\rfloor}$ and $\mathcal{T}_{\lfloor \frac{l+r}{2}\rfloor+1,r}$ and 
\end{itemize}
Note that each vertex of $\mathcal{T}_{l,r}$ is labeled by the set of its descendant leaves.
\end{definition}

\ifdefined\arxiv
\begin{figure}[b]
\centering
\begin{tikzpicture}

\node[circle,draw,color=black,label=above:$\{1\}$] (LUUA) at (-2.0,1.5) {};
\node[circle,draw,color=black,label=above:$\{2\}$] (LUDA) at (-2.0,0.5) {};
\node[circle,draw,color=black,label=above:$\{3\}$] (LDUA) at (-2.0,-0.5) {};
\node[circle,draw,color=black,label=above:$\{4\}$] (LDDA) at (-2.0,-1.5) {};
\node[circle,draw,color=black,fill] (LUB) at (-3.0,1.0) {};
\node[circle,draw,color=black,fill] (LDB) at (-3.0,-1.0) {};
\node[circle,draw,color=black,label=above:{$\{1,2\}$}] (LUA) at (-4.0,1.0) {};
\node[circle,draw,color=black,label=above:{$\{3,4\}$}] (LDA) at (-4.0,-1.0) {};
\node[circle,draw,color=black,fill] (LB) at (-5.0,0.0) {};
\node[circle,draw,color=black,label=above:{$\{1,2,3,4\}$}] (LA) at (-6.0,0.0) {};

\path[thick] (LA) edge (LB);
\path[circle,draw] (LUA) edge (LB);
\path[circle,draw] (LDA) edge (LB);
\path[thick] (LUA) edge (LUB);
\path[circle,draw] (LUUA) edge (LUB);
\path[circle,draw] (LUDA) edge (LUB);
\path[thick] (LDA) edge (LDB);
\path[circle,draw] (LDUA) edge (LDB);
\path[circle,draw] (LDDA) edge (LDB);
\end{tikzpicture}
\begin{align*}
    \begin{bmatrix}
        1&1&1\\
        &1&&1&1\\
        &&1&&&1&1
    \end{bmatrix}
\end{align*}

\caption{\label{fig:submatrix:tree}$\cT_{1,4}$ and its biadjacency matrix}
\end{figure}
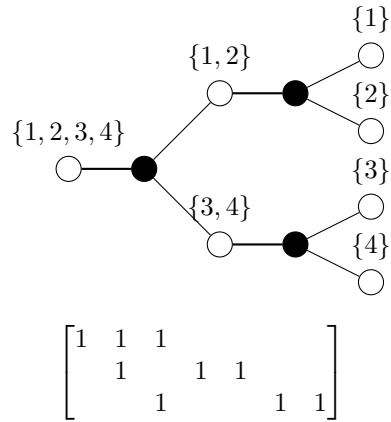
\else
\begin{figure}[t]
\centering
\begin{tabular}{cc}
\begin{tikzpicture}

\node[circle,draw,color=black,label=above:$\{1\}$] (LUUA) at (-2.0,1.5) {};
\node[circle,draw,color=black,label=above:$\{2\}$] (LUDA) at (-2.0,0.5) {};
\node[circle,draw,color=black,label=above:$\{3\}$] (LDUA) at (-2.0,-0.5) {};
\node[circle,draw,color=black,label=above:$\{4\}$] (LDDA) at (-2.0,-1.5) {};
\node[circle,draw,color=black,fill] (LUB) at (-3.0,1.0) {};
\node[circle,draw,color=black,fill] (LDB) at (-3.0,-1.0) {};
\node[circle,draw,color=black,label=above:{$\{1,2\}$}] (LUA) at (-4.0,1.0) {};
\node[circle,draw,color=black,label=above:{$\{3,4\}$}] (LDA) at (-4.0,-1.0) {};
\node[circle,draw,color=black,fill] (LB) at (-5.0,0.0) {};
\node[circle,draw,color=black,label=above:{$\{1,2,3,4\}$}] (LA) at (-6.0,0.0) {};

\path[thick] (LA) edge (LB);
\path[circle,draw] (LUA) edge (LB);
\path[circle,draw] (LDA) edge (LB);
\path[thick] (LUA) edge (LUB);
\path[circle,draw] (LUUA) edge (LUB);
\path[circle,draw] (LUDA) edge (LUB);
\path[thick] (LDA) edge (LDB);
\path[circle,draw] (LDUA) edge (LDB);
\path[circle,draw] (LDDA) edge (LDB);
\end{tikzpicture}
& 
$
    \begin{bmatrix}
        1&1&1\\
        &1&&1&1\\
        &&1&&&1&1
    \end{bmatrix}
$ 
\end{tabular}
\caption{\label{fig:submatrix:tree}$\cT_{1,4}$ and its biadjacency matrix}
\end{figure}
\fi

\noindent
\begin{minipage}{\textwidth}
We write $\cL_{l,r}$ for the set of vertex-labels used in $\cT_{l,r}$.
\begin{definition}\label{def:submatrix:treelabels}
    Let $\cL_{l,r}$ be the labels of vertices in $\cT_{l,r}$. Observe that
    \begin{align*}
        \cL_{l,r}=
        \begin{cases}
            \{\{l,\ldots,r\}\}\cup \cL_{l,\lfloor\frac{l+r}{2}\rfloor}\cup \cL_{\lfloor\frac{l+r}{2}\rfloor+1,r}&l<r\\
            \{\{l\}\}&l=r 
        \end{cases}
    \end{align*}
\end{definition}
\end{minipage}
In addition to the binary-trees, we also have additional ``switch-vertices'' at the bottom, used to ``switch off'' vertices. 
The graph in \Cref{def:submatrix:graphgadget} represents the two trees and the additional switch-vertices at the bottom. 
The vectors $a,b \in (\cL_{1,n})^k$ in \Cref{def:submatrix:graphgadget} represent which $k$ vertices of the tree are currently switched off, i.e., which tree vertices are connected to switch-vertices.
In particular, $a_i$ is the vertex switched off by connecting it to the $i$th switch-vertex.
\begin{definition}\label{def:submatrix:graphgadget}
For $a,b\in (\cL_{1,n})^k$, where $0\le k\le n$, let $B_n(a,b)$
be the graph consisting of
$\mathcal{T}_{1,n}$ and a copy $\mathcal{T}_{1,n}'$ and an additional $2n$ vertices $u_1,\ldots,u_n,v_1,\ldots,v_n$
and additional edges $(u_1,\hat{a}_1),\ldots,(u_k,\hat{a}_k)$, $(v_1,\hat{b}_1'),\ldots,(v_k,\hat{b}_k'), (u_{k+1},v_{k+1}),\ldots,(u_n,v_n)$, where $\hat{a}_k$ is the vertex in $\mathcal{T}_{1,n}$ with label $a_k$, and $\hat{b}_k'$ is the vertex in $\mathcal{T}_{1,n}'$ with label $b_k$.
\end{definition}

The biadjacency matrix of this graph will later be used as our matrix $\mB$ in \Cref{lem:blockmatrixdeterminant}. Here we want to the first $n$ rows/columns of $\mB$ correspond to the leafs of the trees.

\begin{definition}\label{def:submatrix:matrixgadget}
For $a,b \in (\cL_{1,n})^k$, let $\mB_n(a,b)\in\R^{(4n-2)\times(4n-2)}$ be the biadjacency matrix of $B_n(a,b)$, where the leaves of $\mathcal{T}_{1,n}$ are the first $n$ vertices on the left side, and the leaves of $\mathcal{T}_{1,n}'$ are the first $n$ vertices on the right side.
\end{definition}

The binary-tree and switch-vertices will be used to perform a binary-search.
Since graph $B$ has two trees, it will be useful to also define notation where we consider only a single tree with a few switched off vertices. (Observe that in the following definition, only the switch-vertices attached to the tree exist. There are no switch-vertices that are not attached to some tree-vertex.)
Further, as \Cref{lem:blockmatrixdeterminant} argues about $\mB_{-T,-S}$ (i.e., rows/column with index in $T$ and $S$ removed), we also want notation that describes the corresponding vertex deletion to the tree.

\begin{definition}
    For $S\subseteq\{l,\ldots,r\}$ and $a_1,\ldots,a_k\in\cL_{l,r}$, let
    Let $\cT_{l,r}^{a_1,\ldots,a_k}$ denote the tree $\cT_{l,r}$ with $k$ switch-vertices attached to vertices with labels $a_1,\ldots,a_k$ respectively, and let $\cT_{l,r}^{a_1,\ldots,a_k}(-S)$ be $\cT_{l,r}^{a_1,\ldots,a_k}$ with vertices with label $\{i\}$ for $i\in S$ removed.
\end{definition}

The property $\det(\mB)\neq 0$ corresponds to graph $B$ having a perfect matching.
The following \Cref{lem:submatrix:perfectmatching} describes under which conditions our tree has a perfect matching.

\begin{lemma}\label{lem:submatrix:perfectmatching}
Suppose $i_1,\ldots,i_k\in\{l,\ldots,r\}$ Then, $\cT_{l,r}^{\{i_1\},\ldots,\{i_k\}}(-S)$ has a perfect matching iff $S=\{l,\ldots,r\}\setminus\{i_1,\ldots,i_k\}$ and $|S|=r+l-1-k$.
\end{lemma}

\begin{proof}
    Suppose $\cT_{l,r}^{\{i_1\},\ldots,\{i_k\}}(-S)$ has a perfect matching. Then, $i_1,\ldots,i_k\notin S$, as otherwise their switch-vertices will be isolated. Furthermore, $i_1,\ldots,i_k$ must be distinct, as otherwise the neighborhood of the $k$ switch-vertices will have fewer than $k$ vertices, which violates Hall's condition. By the requirement that both sides of the bipartition have the same number of vertices, we get that $|S|=r-l+1-k$. Thus, $S=\{l,\ldots,r\}\setminus\{i_1,\ldots,i_k\}$.

    Conversely, suppose $S=\{l,\ldots,r\}\setminus\{i_1,\ldots,i_k\}$ where $|S|=r-l+1-k$, so that $i_1,\ldots,i_k$ are distinct. We can obtain a perfect matching in $\cT_{l,r}^{\{i_1\},\ldots,\{i_k\}}(-S)$ by matching $\{i_1\},\ldots,\{i_k\}$ to its switch-vertex, and the remaining labeled vertices to their child in $\cT_{l,r}$.
\end{proof}

\Cref{lem:submatrix:perfectmatching} concerns the case where the leafs of the tree are connected to switch-vertices.
To facilitate the binary-search, we also want to connect one other tree-vertex to a switch-vertex.
The following \Cref{lem:submatrix:perfectmatching2} states under which conditions such a graph has a perfect matching.

\begin{lemma}\label{lem:submatrix:perfectmatching2}
Suppose $i_1,\ldots,i_{k-1}\in\{l,\ldots,r\}$ and $I\in\cL_{l,r}$. Then, $\cT_{l,r}^{\{i_1\},\ldots,\{i_{k-1}\},I}(-S)$ has a perfect matching iff $S=\{l,\ldots,r\}\setminus\{i_1,\ldots,i_{k-1},i'\}$ for some $i'\in I$ such that $i_1,\ldots,i_{k-1},i'$ are distinct.
\end{lemma}

\begin{proof}
We will prove the lemma by induction on $r-l$. Let $m=\lfloor \frac{l+r}{2}\rfloor$.

\subparagraph{Perfect Matching $\Rightarrow$ condition on $S$:}

Suppose a perfect matchings exists.
We must have $i_1,\ldots,i_{k-1}$ distinct, as otherwise the neighborhood of the $k-1$ switch-vertices will have fewer than $k-1$ vertices, which violates Hall's condition.

\begin{itemize}
\item Suppose that $I\in \cL_{m+1,r}$.

The vertex labeled $\{l,\ldots,r\}$ must be matched to its child, and so 
$$\cT_{l,m}^{\cL_{l,m}\cap \{\{i_1\},\ldots,\{i_{k-1}\},I\}}(-S\cap \{l,\ldots,m\}) \text{ and}$$
$$\cT_{m+1,r}^{\cL_{m+1,r}\cap \{\{i_1\},\ldots,\{i_{k-1}\}\}}(-S\cap \{m+1,\ldots,r\})$$
must have perfect matchings.

By \Cref{lem:submatrix:perfectmatching}, $S\cap \{l,\ldots,m\}=\{l,\ldots,m\}\setminus\{i_1,\ldots,i_{k-1}\}$.

By the inductive hypothesis, $S\cap \{m+1,\ldots,r\}=\{m+1,\ldots,r\}\setminus\{i_1,\ldots,i_{k-1},i'\}$ for some $i'\in I\setminus\{i_1,\ldots,i_{k-1}\}$.

Thus, $S=(S\cap \{l,\ldots,m\})\cup(S\cap \{m+1,\ldots,r\})=\{l,\ldots,r\}\setminus\{i_1,\ldots,i_{k-1},i'\}$ where $i_1,\ldots,i_{l-1},i'$ are distinct.

\item The case where $I\in \cL_{l,m}$ is similar. 

\item Now, suppose $I=\{l,\ldots,r\}$. Then, the vertex labeled $I$ must be matched to its switch-vertex, and its child matched to either the vertex labeled $\{l,\ldots,m\}$ or $\{m+1,\ldots,r\}$.

\begin{itemize}
\item In the first case, both
$$\cT_{l,m}^{\cL_{l,m}\cap \{\{i_1\},\ldots,\{i_{k-1}\},\{l,\ldots,m
\}\}}(-S\cap \{l,\ldots,m\}) \text{ and}$$
$$\cT_{m+1,r}^{\cL_{m+1,r}\cap \{\{i_1\},\ldots,\{i_{k-1}\}\}}(-S\cap \{m+1,\ldots,r\})$$
have perfect matchings.

By the inductive hypothesis,
$S\cap \{l,\ldots,m\}=\{l,\ldots,m\}\setminus\{i_1,\ldots,i_{k-1},i'\}$ for some $i'\in \{l,\ldots,m\}\setminus\{i_1,\ldots,i_{k-1}\}$.

By \Cref{lem:submatrix:perfectmatching}, $S\cap \{m+1,\ldots,r\}=\{m+1,\ldots,r\}\setminus\{i_1,\ldots,i_{k-1}\}$.

Thus, $S=\{l,\ldots,r\}\setminus \{i_1,\ldots,i_{k-1},i'\}$ for $i'\in\{l,\ldots,m\}\setminus \{i_1,\ldots,i_{k-1}\}$.

\item In the second case, both
$$\cT_{l,m}^{\cL_{l,m}\cap \{\{i_1\},\ldots,\{i_{k-1}\}\}}(-S\cap \{l,\ldots,m\}) \text{ and}$$
$$\cT_{m+1,r}^{\cL_{m+1,r}\cap \{\{i_1\},\ldots,\{i_{k-1}\},\{m+1,\ldots,r
\}\}}(-S\cap \{m+1,\ldots,r\})$$
have perfect matchings.

Again, by \Cref{lem:submatrix:perfectmatching} and the inductive hypothesis, $S=\{l,\ldots,r\}\setminus \{i_1,\ldots,i_{k-1},i'\}$ for $i'\in\{m+1,\ldots,r\}\setminus \{i_1,\ldots,i_{k-1}\}$.
\end{itemize}
\end{itemize}
In all cases we have $S=\{l,\ldots,r\}\setminus\{i_1,\ldots,i_{k-1},i'\}$ for $i'\in I\setminus \{i_1,\ldots,i_{k-1}\}$.
\subparagraph{Condition on $S$ $\Rightarrow$ Perfect Matching:}

Now, we show the converse. Suppose $S=\{l,\ldots,r\}\setminus\{i_1,\ldots,i_{k-1},i'\}$ for some $i'\in I$ such that $i_1,\ldots,i_{k-1},i'$ are distinct. We can explicitly construct a perfect matching of $\cT_{l,r}^{\{i_1\},\ldots,\{i_{k-1}\},I}(-S)$ as follows:

Take the symmetric difference of the unique perfect matching of $\cT_{l,r}(-\{l,\ldots,r\})$ with the path from $I$ to $\{i'\}$, and match $\{i_1\},\ldots,\{i_{k-1}\},I$ to their switch vertices.
\end{proof}

\begin{lemma}\label{lem:submatrix:uniquematching}
    $B_n((\{i_1\},\ldots,\{i_{k-1}\},I),(\{j_1\},\ldots,\{j_{k-1}\},J))$ has a matching covering all vertices 
    except $T$ on the left and $S$ on the right iff $T=[n]\setminus \{i_1,\ldots,i_{k-1},i'\}$ and $S=[n]\setminus \{j_1,\ldots,j_{l-1},j'\}$ for some $i'\in I$ and $j'\in J$, and zero such matchings otherwise.
\end{lemma}

\begin{proof}%

Note that $B((a_1,\ldots,a_k),(b_1,\ldots,b_k))$ is isomorphic to a disjoint union of $\cT_{1,n}^{a_1,\ldots,a_k}$ and $\cT_{1,n}^{b_1,\ldots,b_k}$ and $n-k$ isolated edges.

$B((a_1,\ldots,a_k),(b_1,\ldots,b_k))$ has a perfect matching covering all vertices except $T$ on the left and $S$ on the right iff $\cT_{1,n}^{a_1,\ldots,a_k}(-T)$ and $\cT_{1,n}^{b_1,\ldots,b_k}(-S)$ both have perfect matchings.
\end{proof}

The matrix $\mB_n((\{i_1\},\ldots,\{i_{k-1}\},I),(\{j_1\},\ldots,\{j_{k-1}\},J))_{-T,-S}$ is the biadjacency matrix of the graph \mbox{$B_n((\{i_1\},\ldots,\{i_{k-1}\},I),(\{j_1\},\ldots,\{j_{k-1}\},J))$} after removing vertices $S$ and $T$.
Further, since that graph is a forest, if it has a perfect matching, then it has at most one perfect matching. 
Thus $$\det(\mB_n((\{i_1\},\ldots,\{i_{k-1}\},I),(\{j_1\},\ldots,\{j_{k-1}\},J))_{-T,-S})\ne 0$$ if and only if the graph has a perfect matching, even without replacing the non-zero entries with random field elements. Schwartz-Zippel-Lemma is not needed. Thus we get \Cref{lem:submatrix:Bdeterminant}.

\begin{corollary}\label{lem:submatrix:Bdeterminant}
    $\det(\mB_n((\{i_1\},\ldots,\{i_{k-1}\},I),(\{j_1\},\ldots,\{j_{k-1}\},J))_{-T,-S})\ne 0$ iff $T=[n]\setminus\{i_1,\ldots,i_{k-1},i'\}$ and $S=[n]\setminus\{j_1,\ldots,j_{l-1},j'\}$ for some $i'\in I$ and $j'\in J$.
\end{corollary}

\subsection{Maintaining Maximum Full-Rank Submatrix}
\label{sec:submatrix:combine}

In this section, we prove our reduction \Cref{thm:submatrixreduction} for maintaining a maximum full-rank submatrix $\mA_{I,J}$ ($I,J\subset[n]$) of dynamic matrix $\mA$.
The idea is to greedily add row/column indices to $I,J$, where the correct indices are located via a binary-search.
First, observe that such a greedy procedure works, as proven in \Cref{lem:submatrix:greedy}.

\begin{lemma}\label{lem:submatrix:greedy}
    Suppose $\rank(\mA)>k$. Then, for any $k\times k$ full-rank submatrix of $\mA$, there is some row and column that can be added to make a $(k+1)\times(k+1)$ full-rank submatrix.
\end{lemma}
\begin{proof}
    Suppose $\rank(\mA_{I,J})=k$ for $|I|=|J|=k$.
    
    $I$ indexes $k$ linearly independent rows of $\mA$, so there is some $i\notin I$ such that $I\cup\{i\}$ indexes $k+1$ linearly independent rows. In other words,
    \begin{align*}
        \rank(\mA_{I\cup\{i\},[n]})=k+1
    \end{align*}
    Now, $J$ indexes $k$ linearly independent columns of $\mA_{I\cup\{i\},[n]}$, so there must be some $j$ such that $J\cup\{j\}$ indexes $k+1$ linearly independent columns of $\mA_{I\cup\{i\},[n]}$. Then,
    \begin{align*}
        \rank(\mA_{I\cup\{i\},J\cup\{j\}})=k+1
    \end{align*}
    as desired.
\end{proof}

Next, we show how to find new row/column indices via binary-search, when having access to a dynamic algorithm that detects when matrix $\mH$ (\Cref{lem:blockmatrixdeterminant}) has $\det(\mH)=0$.

\begin{lemma}\label{lem:submatrix:findindices}
    Suppose $\det(\mA_{\{i_1,\ldots,i_k\},\{j_1,\ldots,j_k\}})\ne 0$. We can find $(i',j')$ such that 
    $$\det(\mA_{\{i_1,\ldots,i_k,i'\},\{j_1,\ldots,j_k,j'\}})\ne 0,$$
    or determine that $\rank(\mA)=k$, in $O(\log n)$ calls to a dynamic singularity detection algorithm.
\end{lemma}

\begin{proof}

\begin{algorithm2e}
    \caption{Augment via binary search}
    \label{alg:augment_via_binary_search}
    \SetKwInOut{Input}{Input}
    \SetKwInOut{Output}{Output}
    \Input{$\mA$ and $i_1,\ldots,i_k,j_1,\ldots,j_k$ such that $\det(\mA_{\{i_1,\ldots,i_k\},\{j_1,\ldots,j_k\}})\ne 0$}
    \Output{$(i',j')$ such that $\det(A_{\{i_1,\ldots,i_k,i'\},\{j_1,\ldots,j_k,j'\}})\ne 0$, or that none exist}
    \LinesNumbered %
    \SetAlgoLined %
    \If{$\det(\mH((i_1,\ldots,i_k,\{1,\ldots,n\}),(j_1,\ldots,j_k,\{1,\ldots,n\})))=0$}{
        \KwRet{none exist}
    }
    $l \gets 1, r \gets n$\;
    \While{$l<r$}{
        \If{$\det(\mH((i_1,\ldots,i_k,\{l,\ldots,\lfloor\frac{l+r}{2}\rfloor\}),(j_1,\ldots,j_k,\{1,\ldots,n\})))=0$}{
            $l \gets \lfloor\frac{l+r}{2}\rfloor+1$\;
        }\Else{
            $r \gets \lfloor{\frac{l+r}{2}}\rfloor$\;
        }
    }
    $l' \gets 1, r' \gets n$\;
    \While{$l'<r'$}{
        \If{$\det(\mH((i_1,\ldots,i_k,\{l,\ldots,r\}),(j_1,\ldots,j_k,\{l',\ldots,\lfloor\frac{l'+r'}{2}\rfloor\})))=0$}{
            $l' \gets \lfloor\frac{l'+r'}{2}\rfloor+1$\;
        }\Else{
            $r' \gets \lfloor{\frac{l'+r'}{2}}\rfloor$\;
        }
    }
    \KwRet{$(l,l')$}
\end{algorithm2e}

    See \Cref{alg:augment_via_binary_search}. 
    Here matrix $\mH((i_1,...,i_k,I),j_1,...,j_k,J))$ is the matrix $\mH$ from \Cref{lem:blockmatrixdeterminant} where we use $\mB((i_1,...,i_k,I),(j_1,...,j_k,J))$ (\Cref{def:submatrix:matrixgadget}).

    By \Cref{lem:blockmatrixdeterminant}, $\det(\mH)\neq0$ iff there is $S,T\subset[n]$ with $\det(\mA_{-S,-T}) \neq 0$ and $\det(\mB_{-T,-S})\neq 0$.
    \Cref{lem:submatrix:Bdeterminant} states that $\det(\mB_{-T,-S})\neq 0$ happens iff 
    $$
    T = [n] \setminus \{i_1,...,i_k, i'\},\quad S=[n] \setminus \{j_1,...,j_k, j'\}
    $$
    for some $i'\in I, j'\in J$.
    Thus $\det(\mH)\neq0$ iff there is $i'\in I,j'\in J$ with $0\neq \det(\mA_{-S,-T}) = \det(\mA_{\{i_1,...,i_k,i'\},\{j_1,...,j_k,j'\}})$.
    Hence we can binary-search over sets $I,J$ to find such $i',j'$, which is precisely what we do in \Cref{alg:augment_via_binary_search}.

    Observe that each iteration changes only $O(1)$ entries of $\mB$ to relink the switch-vertices.
    Further, it suffices to have a data structure that detects for the first time when $\det(\mH)=0$, because we can simply revert the previous update.
    The matrix $\mH$ initially has $\det(\mH)\neq 0$ when initializing with $\mB((i_1,...,i_k),(j_1,...,j_k))$.
\end{proof}

We can now prove our main reduction \Cref{thm:submatrixreduction}.
\thmSubmatrix*

\begin{proof}%

We maintain the invariant that $\det(\mA_{\{i_1,\ldots,i_k\},\{j_1,\ldots,j_k\}})\ne 0$ and no larger submatrix has nonzero determinant. 
Further, we always run a dynamic singularity algorithm on matrix $\mM$ (as defined in \Cref{lem:blockmatrixdeterminant}) using $\mB_n(I,J)$. When $\mA_{I,J}$ is full-rank, then $\mM$ is also full-rank.
\subparagraph{Initialization}
To initialize, first find initial maximum sets $I,J$ with $\det(\mA_{I,J})\neq 0$ in $O(n^\omega)$ time.
Then construct $\mM$ and initialize the dynamic singularity data structure.
Thus initialization takes $O(P(n)+n^\omega)$ time.

\subparagraph{Update}

If we update entry $(i,j)$ and $\det(\mA_{I,J})$ becomes zero (implying $i\in I$ and $j\in J$), then $\det(\mM)$ becomes zero. When we detect this, revert the update, remove $i,j$ from $I$ and $J$ (update $\mM$ accordingly), then update $\mA$ (and $\mM$) again. We now have $\det(\mM)\neq 0$ since $\det(\mA_{I,J})\neq 0$ for the new smaller $I,J$.
We perform one binary-search via \Cref{lem:blockmatrixdeterminant} to see if maybe new indices can be added to $I,J$.
If so, add the indices to $I,J$ and update $\mM$.

If we update entry $(i,j)$ and $\det(\mA_{I,J})$ remains nonzero, we also perform one binary-search via \Cref{lem:blockmatrixdeterminant} to see if maybe new indices can be added to $I,J$.
If so, add the indices to $I,J$ and update $\mM$.

Overall, we perform at most $O(\log (n))$ updates to the dynamic singularity algorithm by \Cref{lem:submatrix:findindices}. So the update time is $O(U(n)\log n)$.

\end{proof}

\section{Maintaining Small Rank \& Basis}
\label{sec:smallrank}

In this section we prove the following reduction, allowing any dynamic rank data structure, whose time complexity scales in the dimension, to instead scale in the rank.
Together with dynamic rank data structure of \cite{BrandNS19,Sankowski05} (with $O(n^\omega)$ preprocessing, $O(n^\bns)$ update time for entry-updates, and $O(n^\columnUpdate)$ for column-updates), this then implies \Cref{thm:main_rank}.

\thmSmallRankReduction*

The proof is split into 2 parts. First in \Cref{sec:rank:boundedrank}, we show that one can maintain $\min(\rank(\mA), k)$ of a dynamic matrix $\mA$ and input parameter $k$.
The time complexity of that data structure will scale in $k$.
Then in \Cref{sec:rank:general} we prove \Cref{{thm:rankreduction}} by extending the previous data structure to no longer require the input parameter $k$.
\Cref{sec:rank:basis} then extends the result to maintain a basis in addition to the rank, leading to \Cref{thm:main_basis}.

\subsection{Maintaining Bounded Rank}
\label{sec:rank:boundedrank}

We wish to prove the following theorem, which maintains the rank up to an input parameter $k$.

\begin{theorem}
\label{thm:boundedrank}
Assume there is a dynamic algorithm that maintains the rank of a dynamic matrix $\mA\in\F^{n\times n}$. Let $O(U(n,z))$ be the update time for changing $z$ entries in a column of $\mA$, and let $P(n)$ be the preprocessing time.

Then there exists constant $k_0\in\N$ and a randomized data structure with the following operations: 
\begin{itemize}
    \item \textsc{Initialize}$(\mA\in\F^{n\times n},k>k_0)$: Initializes the data structure in $O(\nnz(\mA)\log n)$ time. The data structure is initially in ``\emph{inactive}'' state.
    \item \textsc{Activate}()/\textsc{Deactivate}(): Changes the active state in $O(P(k)\log n)$ time, and returns w.h.p.~$\min(k, \rank(\mA))$.
    \item \textsc{Update}$(i\in[n],v\in\F^n)$ Set $\mA \gets \mA + v e_i$ (i.e., change the $i$th column by adding $v$ to it).
    
    If ``active,'' the data structure returns $\min(k, \rank(\mA))$ in $O(U(k, \nnz(v))\log n)$ worst-case time. If ``inactive,'' nothing is returned and the update takes only $O(\nnz(v) \log n)$ time.
\end{itemize}
The data structure is randomized and correct with high probability, and works against an adaptive adversary.
\end{theorem}

\begin{proof}
Let $\mM,\mN$ be two matrices of \Cref{lem:rankprojection} sampled for $k,n$.

Our data structure maintains $\mM^\top\mA\mN$ throughout all updates, regardless of active/inactive state.
If in active state, we maintain the rank of $\mM^\top\mA\mN$, which by \Cref{lem:rankprojection} has 
$$\rank(\mM^\top \mA \mN) = \min(k, \rank(\mM^\top \mA)) = \min(k, \rank(\mA)).$$
Observe that $\mM^\top \mA \mN$ is a $O(k)\times O(k)$ matrix, hence the time complexity of maintaining its rank will scale in $k$ rather than $n$.

\paragraph{Initialization} 
We sample $\mM,\mN$ in $O(n)$ time. Because of their sparsity (i.e., $2$ non-zero entries per row), computing $\mM^\top\mA\mN$ takes $O(\nnz(\mA)$ time. Without loss of generality $\nnz(\mA) \ge n$, so the time for initialization is $O(\nnz(\mA))$.

\paragraph{Activation} 
We initialize the assumed dynamic rank algorithm on $\mM^\top \mA \mN$ in $O(P(k))$ time.

\paragraph{Update} 
Consider what happens when we update our matrix $\mA$ by changing a column or single entry. We may express this update as $\mA \gets \mA + v e_i^\top$ where $v\in\F^n$ and $e_i$ is a standard unit vector.
In case of column updates, $v$ is added to the $i$th column of $\mA$. In case of entry updates, $v$ is a sparse vector with a single non-zero entry.

We now have
$$\mM^\top (A + v e_i^\top) \mN = \mM^\top A \mN + (\mM^\top v) (e_i^\top \mN)$$
By $\mM,\mN$ having only $2$ non-zero entries per row, we can compute $u:=\mM^\top v$ in $O(\nnz(v))$ time.
Observe that $u$ has only $O(\nnz(v))$ non-zero entries. Vector $e_i^\top \mN$ is a single row of $\mN$, thus having $2$ non-zero entries. Thus the outer product $u (e_i^\top \mN)$ results can be computed in $O(\nnz(v))$ time.

This implies we can maintain $\mM^\top \mA \mN$ in $O(\nnz(v))$ time per column update, where $\nnz(v)$ is the number of entries of $\mA$ that change during the column update.

Further, each column update to $\mA$ changing $\nnz(v)$ entries, results in 2 column updates to $\mM^\top \mA \mN$ changing $O(\nnz(v))$ entries per column.

\paragraph{Additional Update while Active}

If the data structure is in active state, we pass the two column updates to the assumed dynamic rank data structure.
As argued above, each column update changes $O(\nnz(v))$ entries, thus taking $O(U(r,\nnz(v))$ time to maintain the rank.

\paragraph{Failure Probability}

We can assume without loss of generality that field $\F$ has size $\poly(n)$, otherwise we use a field extension \cite[Lemma 2.1]{CheungKL13}.
Thus by \Cref{lem:rankprojection}, we have with constant probability that $\rank(\mM^\top \mA \mN) = \min(k,\rank(\mA))$.
We can boost this by running $O(\log n)$ instances of our data structure and returning the maximum rank.

\end{proof}

\subsection{Maintain Unbounded Rank}
\label{sec:rank:general}

We now use the data structure of \Cref{thm:boundedrank} to prove \Cref{thm:rankreduction}.

\thmSmallRankReduction*

Note that we are no longer returning $\min(k, \rank(\mA))$ but the exact $\rank(\mA)$.

\begin{proof}%

Our data structure initializes $O(\log n)$ data structure of \Cref{thm:boundedrank} for parameter $k=2^i$ for $i=\log(k_0),...,\log n$.
Throughout all updates, our data structure maintains a parameter $j$ satisfying the following invariants:
\begin{align}
2^j/4 \le \rank(\mA) \le 2^{j} \label{eq:invariant:j}\\
\text{For any $i$ with $2^i/2 \le \rank(\mA) \le 2^{i}$, the $i$th data structure is \emph{active}.} \label{eq:invariant:active}\\
\text{For $i>j+1$ and $i<j-2$, the $i$th data structures are \emph{inactive}.}\label{eq:invariant:inactive}
\end{align}
Thus the $j$th data structure returns the correct rank of $\mA$.
Further, the update time will be bounded wrt.~$\rank(\mA)$ rather than $n$, because only for $2^i \approx \rank(\mA)$ will the $i$th data structures be active.

Observe that for $i=j-2,j-1$ and $i=j+1$, the data structure can be active, but does not have to be active. Only at edge cases $\rank(\mA)=2^j/2$ or $\rank(\mA)=2^j$ must the $i=j-1$ or $i=j+1$ data structure be active.
This property will later be used to facilitate worst-case update time.

\paragraph{Initialization}
We initialize the $O(\log n)$ data structures in $O(\nnz(\mA) \log n)$ time.
We then compute the rank of $\mA$ and pick $j$ accordingly. We activate the data structures for $i=j-2,..,j+1$ in $O(P(\rank(\mA)))$ time.

Computation of the rank can for example be done by activating all data structure for increasing $i$ until the returned rank does not increase any further. Then deactivating the data structures for $i<j-2$ again.
This takes $O(\sum_{i=0}^{\log\rank(\mA)}P(2^i))=O(P(\rank(\mA)))$ time.

\paragraph{Update}
We pass the updates to all data structures, which takes $O(z \log n)$ time for the inactive ones, and $O(U(2^i, z)) = O(U(\rank(\mA, z))$ for the active data structures. Here $2^i = \Theta(\rank(\mA))$ holds for active data structures by invariant \eqref{eq:invariant:j} and \eqref{eq:invariant:inactive}.

To maintain the invariant, we proceed as follows:
Since we perform column updates, the rank of $\mA$ can change by at most 1 per update.
Thus we know after each update whether $2^j = \rank(\mA)$, because the rank is correctly maintained by the $j$th data structure.

If $\rank(\mA)=2^j$, then we set $j\gets j+1$ to maintain \eqref{eq:invariant:j}.
To maintain invariants \eqref{eq:invariant:inactive} and \eqref{eq:invariant:active}, we deactivate the $(j-2)$th data structure and activate the $(j+1)$th data structure.

If instead the rank decreases and we have $\rank(\mA) = 2^{j-1}/4$, then we set $j \gets j-2$ to maintain \eqref{eq:invariant:j}.
Then we deactivate the $i$th data structures for $i>j+1$, and active the $(j-2)$th and $(j-3)$th data structure to maintain \eqref{eq:invariant:inactive} and \eqref{eq:invariant:active}.

The amortized time for activating the data structures is $O(P(\rank(\mA))/\rank(\mA))$ because because it takes $2^j$ rank-increasing updates, or $2^j(3/4)$ rank-decreasing updates until at most two activations occur.

\paragraph{Worst-Case Time}

If we spread the activation cost of the $i$th data structure over $2^i/8$ future updates, we obtain $O(P(2^i)/2^i)$ worst-case time per update rather than spending $O(P(2^i))$ in a single iteration.

If we do this, then after $2^i/8$ updates, the $i$th data structure is active but ``outdated'' (ie. missing the past $2^i/8$ updates). So over the future $2^i/8$ updates, we perform two queued updates per iteration.
After these additional $2^i/8$ updates, the $i$th data structure is ``live,'' i.e., active and represents the current status of input matrix $\mA$.

Observe that this entire process to activate the $i$th data structures takes $2^i/4$ updates.
Since the rank of $\mA$ can change by at most 1 per update, the $i$th data structure has become active soon enough to still maintain invariant \eqref{eq:invariant:active}.

As activations only happen for $2^i = \Theta(\rank(\mA))$, the worst-case time per update becomes
$$
O(z \log n + U(\rank(\mA), z) + P(\rank(\mA))/\rank(\mA)).
$$

\end{proof}

\subsection{Maintaining Column-Basis}
\label{sec:rank:basis}

In this subsection we prove \Cref{thm:basis_reduction}.

\thmBasisReduction*

With the previous dynamic rank data structure of \Cref{thm:main_rank}, \Cref{thm:basis_reduction} implies \Cref{thm:main_basis}.
In particular, we can maintain the column basis in $\tilde O(\rank(\mA)^\columnUpdate + z)$ time.

Before proving the theorem, we first state and prove a small helper-lemma/data structure.

\begin{lemma}\label{lem:Av}
    Given $\mA\in\F^{n\times n}$ and $n$-dimensional vector $v\in\F^n$,
    let $v^{(\ell,k)}$ be the vector with
    $$
    v^{(\ell,k)}_i = \begin{cases}
        v_i & \text{if } i \in (k \cdot 2^\ell, (k+1) \cdot 2^\ell] \\
        0 & \text{otherwise}
    \end{cases}
    $$

    After $O(n^2)$ time preprocessing,
    we can maintain the products $\mA v^{(\ell,k)}$ for all $0 \le \ell \le \log n$ and $0 \le k < n/2^\ell$ per column update to $\mA$. A column update that changes $z$ entries of $\mA$ takes $O(z \log n)$ time.
\end{lemma}

\begin{proof}
    During initialization we compute $\mA v^{(\ell, k)}$ for $\ell=0$ and $k=0,..,n-1$. this takes $O(n^2)$ time.
    Then we compute $\mA v^{(\ell-1, k)} = \mA^{(\ell,2k)}+\mA^{(\ell,2k+1)}$ for all $\ell,k$ which takes
    $\sum_{\ell=0}^{\log n} 2^\ell \cdot O(n) = O(n^2)$ time.

    When a column update to $\mA$ happens, i.e., $\mA \gets \mA + u e_i^\top$, then for each $\ell=0,...,\log n$ there is exactly one $k$ where $\mA v^{(\ell,k)}$ changes.
    The value is given by $(\mA + u e_i^\top) v^{(\ell,k)} = \mA v^{(\ell,k)} + u v_i$. 
    Thus such an update takes $O(\nnz(u) \log n)$ time for all $\ell$.
\end{proof}

\begin{proof}[Proof of \Cref{thm:basis_reduction}]
    We assume existence of a dynamic rank data structure with update time $O(U(n))$ for $n\times n$ input matrix and column-update that changes $z$ entries.
    We start the proof by explaining how to maintain a column-basis of dynamic matrix $\mA$ in update time $O((U(n) + z) \log n)$. Then we improve the time complexity to $\tilde O(P(\rank(\mA)/\rank(\mA) + U(\rank(\mA)) + z)$.

    \paragraph{Initialization}
    During initialization we compute a column basis of input matrix $\mA$.
    Let $\mB \in \F^{n\times n}$ be the column basis, i.e., the first $\rank(\mA)$ columns are filled and the rest is 0 to guarantee square-ness of $\mB$.
    Initialize the assumed dynamic rank data structure on $\mB$.
    We also initialize the $\mA v$-data structure of \Cref{lem:Av} on $\mA$ and a random vector $v\in\F^n$, where each entry is i.i.d.~a uniformly sampled field element from $\F$.
    \paragraph{Column Update}
    We update the $\mA v$-data structure (\Cref{lem:Av}).
    If the updated column is part of the basis (i.e., $\mB$), then we update the $\rank(\mB)$-data structure.
    If the rank went down, we know the updated column has become linearly dependent, so we remove it from $\mB$ and update the $\rank(\mB)$-data structure again.

    Next, we must check if maybe another column of $\mA$ could be added to the basis.
    If there exists a column, then by Schwartz-Zippel-\Cref{lem:schwartzzippel} we have that w.h.p.\footnote{Assuming $\F$ is a field of size $\poly(n)$. If it isn't, then we use a poly-sized field extension of $\F$ instead.}~ $\mA v$ is linearly independent of $\mB$.
    If there exists no column, then $\mA v$ must be linearly dependent to $\mB$.
    Thus we can binary-search for the first linearly independent column in $\mA$ by adding column $\mA v^{(\ell,k)}$ (see definition in \Cref{lem:Av}) to $\mB$ and checking if the rank increased, then removing the column again.
    At the end, $\mB$ contains a new linearly independent column of $\mA$.

    In total, we will perform at most $O(\log n)$ column-updates to the $\rank(\mB)$-data structure, and 1 update to the $\mA v$-data structure.
    Thus the time complexity is $O((U(n) + z) \log n)$.

    \paragraph{Low-Rank Case}
    Consider a random $k\times n$ matrix $\mM$ from \Cref{lem:rankprojection}. Then $\rank(\mM^\top \mB) = \min(\rank(\mB), r)$ with some constant probability. By sampling multiple $\mM_1,...,\mM_{O(\log n)}$ and taking the maximum, we get $\max_i \rank(\mM_i^\top \mB) = \min(\rank(\mB, r)$.
    Thus, while the rank is at most $r$, we can run the assumed dynamic rank data structure on $\mM^\top_i \mB$, and the $\mA v$-data structure on $(\mM_i^\top \mA) v$, for $i=1,...,O(\log n)$.
    This leads to $\tilde O(U(r)+z)$ update time, because changing $z$ entries in $\mA$, changes at most $2z$ entries in each $\mM_i \mB$ and $\mM_i \mA$.

    At last, we repeat the above for $r=2^j$ and $1 \le j \le \log n$. Throughout all updates we maintain all products $\mM^\top \mA$ for all $r$. Of the rank data structures, only those with $\rank(\mA)/2 \le r \le 2\rank(\mA)$ are active, so we can amortize the initialization cost over $\rank(\mA)$ updates.
    For the matrix vector product data structures (\Cref{lem:Av}) we can keep all of them active.
    Thus the update time becomes
    $$\tilde O(P(\rank(\mA))/\rank(\mA) + U(\rank(\mA)) + z).$$

    \paragraph{Adaptivity}
    The output of our dynamic algorithm would w.h.p.~be the same if we always compute $\rank(\mB)$ deterministically and append to $\mB$ the first column from $\mA$ that is linearly independent from $\mB$. Thus the output of our data structure is pseudo-deterministic and works against an adaptive adversary. 
\end{proof}

\subsection{Dynamic Maximum Matching Weight \& Size}

In this subsection, we prove \Cref{thm:main_matching} via standard reductions.

\thmMatching*

The result for maximum matching follows from the well-known reduction to the rank of a Tutte-matrix.

\begin{lemma}[{\cite{tutte1947factorization}}]
    Given graph $G=(V,E)$, let $\mA$ be the symbolic matrix where for each edge $\{u,v\} \in E$,
    $\mA_{u,v} = x_{uv} = -\mA_{v,u}$.
    Then $\rank(\mA)$ is twice the size of the maximum matching in $G$.

    In particular, if $\mA' \in \F^{n\times n}$ is matrix $\mA$ where we replace each $x_{uv}$ with a random element of polynomially-sized field $\F$, then w.h.p.~$\rank(\mA')$ is twice the size of the maximum matching in $G$.
\end{lemma}

The bound for maintaining the size of a matching in $\tilde O(|M|^\bns)$ thus follows directly by maintaining the rank of $\mA'$ via the data structure of \Cref{thm:main_rank} in $\tilde O(\rank(\mA'))$ time.
Each edge update changes 2 entries in $\mA'$, and thus we must perform only two entry-updates to the dynamic rank data structure.

The result for weighted matching follows from \Cref{lem:weightedmatchingreduction}.

\begin{lemma}[{\cite{KaoLST01}}]\label{lem:weightedmatchingreduction}
    Let $G$ be an integer weighted bipartite graph.
    Construct $H$ as follows:
    For $v \in V(G)$ let $w_v$ be the largest incident edge weight. Create $w_v$ many vertices $v_1,...,v_{w_v}$ in $H$.
    For each edge $\{u,v\}\in E(G)$, create edges $\{u_1, v_w\},\{u_2, v_{w-1}\},...,\{u_w, v_1\}$, where $w$ is the weight of edge $\{u,v\}$.

    Then the size of the maximum matching in $H$ is equal to the maximum weight matching in $G$.
\end{lemma}

Let $\mA'$ be the randomized Tutte-matrix of graph $H$.
Then a weight $W$ edge-update to $G$ changes $W$ edges in $H$ and thus $W$ entries in matrix $\mA'$.
The rank of $\mA'$ is the size of the maximum matching in $H$ and thus the weight of the maximum weight matching in $G$.
In particular, we can maintain the maximum weight in $\tilde O(W\cdot \rank(\mA')^\bns) = \tilde O (W\cdot k^\bns)$ time per edge-update, where $W$ is the weight of the updated edge and $k$ is the weight of the maximum weight matching.



\section*{Acknowledgments}

Jan van den Brand is supported
by NSF Award CCF-2338816 and CCF-2504994. 
Daniel J. Zhang is supported by NSF Award CCF-2338816.

\bibliography{ref}
\bibliographystyle{alpha}

\clearpage

\appendix
\section{Deterministic Combinatorial Dynamic Matching}

In this section, we give a deterministic combinatorial algorithm for maintaining a maximum matching $|M|$ on a bipartite graph $G$ in $O(|M|^2)$ per update.

\thmCombinatorial

The algorithm is based on the following observation, which is a variation of the ``core-graphs'' of \cite{GuptaP13}: Let $W\subset V$ be the matched vertices, then any augmenting path of length $\ge 3$ uses (except for the first and last unmatched edge on the augmenting path) only edges from the induced subgraph $G[W]$.
Let $W' \supseteq W$ by adding for each $w\in W$ an unmatched neighbor to the set (if $w$ has an unmatched neighbor). 
Then there is an augmenting path of length $\ge 3$ in $G$ if and only if there is an augmenting path of length $\ge 3$ in $G[W']$, obtained by simply replacing the first and last vertex on the augmenting path. By $|W'|\le 2 |M|$, finding an augmenting path in $G[W']$ takes only $O(|M|^2)$ time.

To prove \Cref{thm:main_combinatorial} rigorously, we need to take care of how to maintain the set $W'$.

We describe a helper data structure that supports some basic operations on graphs.

\begin{lemma}\label{lem:graph_data_structure}
    There is a deterministic data structure that supports the following operations on a graph $G$ on $n$ vertices.
    \begin{itemize}
        \item $\texttt{Initialize}()$ the empty graph in $O(n^2)$.
        \item $\texttt{Test}(e)$ that returns if $e\in E(G)$.
        \item $\texttt{Insert}(e)$ that inserts an edge $e\notin E(G)$ in $O(1)$.
        \item $\texttt{Delete}(e)$ that removes an edge $e\in E(G)$ in $O(1)$.
        \item $\texttt{Subgraph}(S)$ that returns $G[S]$ induced by $S\subseteq V(G)$ in $O(|S|^2)$.
        \item $\texttt{ExternalNeighbors}(S)$ that returns a mapping $x:S\to V(G)\cup\{\bot\}$ such that for each vertex $v\in S$, $x(v)\in N_G(v)\setminus S$, or $x(v)=\bot$ if none exist, in $O(|S|^2)$.
    \end{itemize}
\end{lemma}

The main difficulty is the $\texttt{ExternalNeighbors}$ operation. 
We would like to enumerate neighbors in $O(1)$ time each, which suggests using an adjacency list. However, we also wish to check if edges are present in $O(1)$ time, which suggests using an adjacency matrix instead.

\begin{proof}
We will use a hybrid of an adjacency matrix and an adjacency list. It will be convenient to store each edge as two arcs.

We assume the vertices are labeled $1,\ldots,n$ and so can be used to index into an array.

We will store an $n\times n$ matrix $A$ where each entry $A_{uv}$ is either empty or a node indicating the presence of an arc from vertex $u$ to vertex $v$.
We also have $n$ special row nodes and $n$ special column nodes, one for each row and column, respectively.

The nodes in each row will be linked in a circular doubly-linked list, along with their row node. The nodes in each column will be linked in a circular doubly-linked list, along with their column node. In other words, each arc stores the previous and next arc (not necessarily ordered by label) in its row list and the previous and next arc in its column list.

Testing if an edge $uv$ is present can be done by checking if $A_{uv}$ is nonempty.

To insert an edge $uv$, we fill $A_{uv}$ with a node, and insert it into the row list for $u$ using the special row node, and likewise for the column. We do the same for $A_{vu}$.

To delete an edge $uv$, we remove the node $A_{uv}$ from its row and column lists, and set it to empty.

To output $G[S]$, we just look at the $S\times S$ entries in $A$.

To find an entry of $N_G(v)\setminus S$ (if one exists) in $O(|S|)$, we iterate over the neighbors of $v$ (by iterating through the linked list starting at its row node) and skip ones that belong to $S$. We can determine what belongs to $S$ using an auxiliary array of size $n$ that gets temporarily marked with the entries in $S$.

\end{proof}

The following \Cref{lem:locally_augment} formalizes the observation that we can restrict our search to an induced subgraph.
\begin{lemma}\label{lem:locally_augment}
Let $G$ be a bipartite graph. Suppose $M$ is a maximal matching in $G$.

Let $x:S\to V(G)\cup\{\bot\}$ be any mapping such that for each $v\in S$, $x(v)\in N_G(v)\setminus S$, or $x(v)=\bot$ if none exist. Let $W=V(M)\cup\{x(v):v\in V(M)\}\setminus\{\bot\}$.

Then, either $M$ is a maximum matching in $G$, or there exists 
an augmenting path for $M$ in $G[W]$.
\end{lemma}
\begin{proof}
Suppose $M$ is a maximal matching in $G$ which is not a maximum matching.
Let $N$ be a maximum matching in $G$. Since each vertex has degree one or two, $G[M\Delta N]$ consists of a vertex-disjoint union of paths and cycles. Each path and cycle alternates between edges of $M$ and edges of $N$. Since $|N|>|M|$, there must be a component of $G[M\Delta N]$ that has more edges from $N$ than $M$, which is necessarily an augmenting path. Let $P=v_1v_2\ldots v_k$ be this augmenting path, where $k\ge 2$ is even. Note that $V(P)\setminus\{v_1,v_k\}\subseteq V(M)$ as $v_2v_3,\ldots,v_{k-2}v_{k-1}$ are edges in $M$. Also, $v_1,v_k\notin V(M)$.

If $k=2$, then $M+v_1v_2$ would be a larger matching than $M$ in $G$, contradicting maximality.

Therefore, $k\ge 4$ and $v_2,v_{k-1}\in V(M)$. Let $v_1'=x(v_2)$ and $v_k'=x(v_{k-1})$. They cannot be $\bot$ since $v_1\in N_G(v_2)\setminus V(M)$ and $v_k\in N_G(v_{k-1})\setminus V(M)$.
Since $v_1',v_k'\in W\setminus V(M)$, they are distinct from $v_2,\ldots,v_{k-1}\in V(M)$. Furthermore, $v_1\ne v_k$ as they are on different sides of the bipartition of $G$. Hence, $P':=v_1'v_2\cdots v_{k-1}v_k'$ is a path in $G[W]$. It is clearly an augmenting path for $M$ in $G[W]$.
\end{proof}

We prove \Cref{thm:main_combinatorial} as the following \Cref{lem:combinatorial}.

\begin{lemma}[\Cref{thm:main_combinatorial}]\label{lem:combinatorial}%
There is a deterministic data structure that maintains a maximum matching $M$ of a bipartite graph $G$ on $n$ vertices which supports the following operations.
    \begin{itemize}
        \item $\texttt{Initialize}()$ the empty graph in $O(n^2)$.
        \item $\texttt{Insert}(e)$ that inserts an edge $e\notin E(G)$ in $O(|M|^2)$.
        \item $\texttt{Delete}(e)$ that removes an edge $e\in E(G)$ in $O(|M|^2)$.
        \item $\texttt{Matching}()$ that returns a maximum matching of $G$ in $O(|M|)$.
    \end{itemize}
\end{lemma}

\begin{proof}
We will maintain $G$ in an instance of the data structure of \Cref{lem:graph_data_structure}, along with a set of edges $M$.

We will keep the invariant that between operations, $M$ is a maximum matching of $G$.

Observe that if $M$ is a maximal matching in $G$ whose size is at most one fewer than a maximum matching in $G$, we can make it maximum matching in $O(|M|^2)$ as follows.
Let $x\gets \texttt{ExternalNeighbors}(V(M))$. %
Let $W\gets V(M)\cup\{x(v):v\in V(M)\}\setminus\{\bot\}$.
We then construct $G[W]$ explicitly.
By \Cref{lem:locally_augment}, either $M$ is a maximum matching, or it has an augmenting path in $G[W]$. We can find an augmenting path $P$ in $G[W]$ (if one exists) in $O(|W|^2)$ using breadth-first search, and set $M\gets M\Delta P$ to get a maximum matching. Each of these steps can be done in $O(|M|^2)$ time.

To delete an edge $e\notin M$, we simply remove it from $G$. The size of the maximum matching cannot increase, so $M$ remains a maximum matching.

To delete an edge $e\in M$, we remove it from $G$ and from $M$. Now, $M$ is a maximal matching in $G$ whose size is at most one fewer than a maximum matching in $G$.
We now augment as in the preceding discussion.

To insert an edge $uv$ where $u,v\notin V(M)$, we insert it in $G$ and in $M$.

To insert an edge $uv$ where $u\in V(M)$ or $v\in V(M)$, we insert it in $G$. Then, $M$ is a maximal matching in $G$ whose size is at most one fewer than a maximum matching in $G$. We now augment as in the preceding discussion.
\end{proof}

\section{Other Facts about Basis}
\label{app:basis}

The following lemma shows that maintaining a maximum nonsingular submatrix is equivalent to maintaining a row and column basis. Note that the requirement of having bases here is crucial: a linearly independent set of rows and a linearly independent set of columns do not necessarily index a full-rank submatrix (for example, the first row and second column of an identity matrix).
\begin{lemma}\label{lem:submatrix_equiv_basis}
    Let $\mA\in\mathbb{F}^{m\times n}$ be a matrix and $I\subseteq[m],J\subseteq[n]$. Then, $\mA_{I,J}$ is a nonsingular submatrix of size $|I|=|J|=\rank(\mA)$ iff the rows of $\mA$ indexed by $I$ form a row basis and the columns of $\mA$ indexed by $J$ form a column basis.
\end{lemma}
\begin{proof}
    If $\mA_{I,J}$ is nonsingular and $|I|=|J|=\rank(\mA)$, then
    \begin{align*}
        \rank(\mA)=|I|=\rank(\mA_{I,J})\le \rank(\mA_{I,[n]})\le \rank(\mA)
    \end{align*}
    Thus, the rows indexed by $I$ are $\rank(\mA)$ linearly independent rows of $\mA$, and are thus a basis. Similarly, the columns indexed by $J$ are $\rank(\mA)$ linearly independent columns of $\mA$, and are thus a basis.

    Conversely, suppose that $I$ indexes a row basis and $J$ indexes a column basis. Then, every row of $\mA$ with index not in $I$ can be written as a linear combination of rows of $\mA$ with index in $I$. In particular, $\rank(\mA_{I,[n]})=\rank(\mA)$. Similarly, every column of $\mA$ with index not in $J$ can be written as a linear combination of columns of $\mA$ with index in $J$. The same holds for $\mA_{I,[n]}$ in place of $\mA$, using the same linear combinations. Hence, $\rank(\mA_{I,J})=\rank(\mA_{I,[n]})=\rank(\mA)$. Since $|I|=|J|=\rank(\mA)$, $\mA_{I,J}$ is nonsingular.
\end{proof}

\begin{lemma}
    Let $\mA\in\F^{V\times V}$ be the Tutte matrix of a graph $G=(V,E)$. Then, the rows of $\mA$ indexed by $I\subseteq V$ form a row basis iff $I$ is the vertex set of a maximum matching in $G$.
\end{lemma}

\begin{proof}
    Suppose the rows of $\mA$ indexed by $I$
form a row basis. Since the Tutte matrix is skew-symmetric, the columns of $\mA$ indexed by $I$ form a column basis. By \Cref{lem:submatrix_equiv_basis}, $\mA_{I,I}$ is nonsingular. Note that $\mA_{I,I}$ is the Tutte matrix of $G[I]$ (the subgraph of $G$ induced by $I$). Since $\mA_{I,I}$ is nonsingular, $G[I]$ must have a perfect matching. That is, there is a matching in $G$ with vertex set $I$. This must be a maximum matching: if $J$ is the vertex set of a larger matching, then $\mA_{J,J}$ is nonsingular and so $\rank(\mA)\ge \rank(\mA_{J,J})=|J|>|I|$, contradicting $I$ indexing a row basis.

Conversely, suppose $I$ is the vertex set of a maximum matching in $G$. Then, $\mA_{I,I}$, being the Tutte matrix of $G[I]$, is nonsingular. Thus, the rows of $\mA$ indexed by $I$ are linearly independent. Let $J\supseteq I$ index a row basis in $\mA$. By the previous paragraph, $J$ is the vertex set of a maximum matching in $G$. Hence, $J=I$.

\end{proof}

\end{document}